\pdfoutput=1

\documentclass[twocolumn]{autart}%
\usepackage{bbm}
\usepackage{amssymb}
\usepackage{amsmath}
\usepackage{amssymb}
\usepackage{mathrsfs}
\usepackage{amsfonts}
\usepackage{graphicx}
\usepackage{graphicx}
\usepackage{booktabs}%
\usepackage{epstopdf}
\UseRawInputEncoding
\providecommand{\U}[1]{\protect \rule{.1in}{.1in}}
\newenvironment{proof}[1][Proof]{\noindent \textbf{#1.} }{\  \rule{0.5em}{0.5em}}
\allowdisplaybreaks[4]
\newtheorem{remark}{Remark}
\newtheorem{theorem}{Theorem}
\newtheorem{lemma}{Lemma}

\newtheorem{corollary}{Corollary}

\allowdisplaybreaks[4]
\begin{document}
\begin{frontmatter}
\title{Extremum seeking in the presence of large delays via time-delay approach to averaging\thanksref{footnoteinfo}}

\thanks[footnoteinfo]{This work was supported by the National Natural Science
Foundation of China (Grant No. 61903102), the Fundamental Research Funds for the Central Universities £¨Grant No. HIT.OCEF.2023007), the Planning and Budgeting Committee (PBC) Fellowship from the Council for Higher Education in Israel, the Israel Science Foundation (Grant No. 673/19) and the C. and H. Manderman Chair at Tel Aviv University.}
\author[a]{Xuefei Yang}\ead{yangxuefei@hit.edu.cn}    
\author[b]{\quad Emilia Fridman}\ead{emilia@tauex.tau.ac.il}
\address[a]{Center for Control Theory and Guidance Technology, Harbin Institute of Technology, Harbin, China.}
\address[b]{School of Electrical Engineering, Tel-Aviv University, Israel.}

\begin{abstract}
In this paper, we study gradient-based classical extremum seeking (ES) for uncertain n-dimensional (nD) static quadratic maps in the presence of known large constant distinct input delays and large output constant delay with a small time-varying uncertainty. This uncertainty may appear due to network-based measurements. We present a quantitative analysis via a time-delay approach to averaging. We assume that the Hessian has a nominal known part and norm-bounded uncertainty, the extremum point belongs to a known box, whereas the extremum value to a known interval. By using the orthogonal transformation, we first transform the original static quadratic map into a new one with the Hessian containing a nominal diagonal part. We apply further a time-delay transformation to the resulting ES system and arrive at a time-delay system, which is a perturbation of a linear time-delay system with constant coefficients. Given large delays, we choose appropriate gains to guarantee stability of this linear system.
 To find a lower bound on the dither frequency for practical stability, we employ variation of constants formula and exploit the delay-dependent positivity of the fundamental solutions of the linear system with their tight exponential bounds. Sampled-data ES in the presence of large distinct input delays is also presented. Explicit conditions in terms of simple scalar inequalities depending on tuning parameters and delay bounds are established to guarantee the practical stability of the ES control systems. We show that given any large delays and initial box, by choosing appropriate gains we can
achieve practical stability for fast enough dithers and small enough uncertainties.
\end{abstract}

\end{frontmatter}

\section{Introduction}

ES is a model-free, real-time on-line adaptive optimization control method.
In 2000, Krstic and Wang gave the first
rigorous stability analysis for an ES system by using averaging and singular
perturbations in \cite{kw00auto}, which laid a theoretical foundation for the
development of ES. Subsequently, a great amount of theoretical and applied
studies on ES are emerging, see e.g. non-local ES control in
\cite{tnm06auto}, Newton-based ES in \cite{gkn12auto,mmb10tac}, ES
via Lie bracket approximation in \cite{dsej13auto,lgke19auto}, stochastic ES
in \cite{lk12book} and ES control for wind
farm power maximization in \cite{eg17ifac}.

The time-delay phenomenon often exists in applications of ES, due to time
needed to measuring and processing of the data
(\cite{mk21auto,okt17tac,ok22book}). The existence of time-delay can cause
performance deterioration and even instability of the ES control system. To
address the challenges of delays in extremum seeking, Oliveira et al. in
\cite{okt17tac} provided classical predictors for delayed gradient and
Newton-based methods. This work was later extended to
single-variable ES with uncertain constant delay in \cite{rokg21ejc}, and with
time-varying delay in \cite{orddk2018acc,rokg21ijacs}. Recently, Malisoff et
al. in \cite{mk21auto} reconsidered the multi-variable ES for static maps with
arbitrarily long time constant delays by using a one-stage sequential
predictor, which can avoid the interference of the integral term appeared in
\cite{okt17tac}. However, these methods provide only qualitative analysis employing the classical averaging theory
in infinite dimensions (see \cite{hl90jiea}),
and cannot suggest quantitative bounds on the dither frequency that preserve
the stability.

Recently, a new constructive time-delay approach to the continuous-time
averaging was presented in \cite{fz20auto} with efficient and quantitative
bounds on the small parameter that ensures the stability. The time-delay
approach to averaging was successfully applied for the quantitative stability
analysis of continuous-time ES algorithms in \cite{zf22auto} and sampled-data
ES algorithms in \cite{zfo22tac} for static quadratic maps by constructing
appropriate Lyapunov-Krasovskii (L-K) functionals. However, the analysis via
L-K method is complicated and the results may be conservative. In our recent
paper \cite{yf22tac} we suggested a robust time-delay approach to ES, where we
presented the resulting time-delay model as a non-delayed one with
disturbances and further employed a variation of constants formula. The latter
can greatly simplify the stability analysis via L-K method, simplify the
conditions and reduce conservatism.

In this paper, we consider the multi-variable ES of uncertain static quadratic
maps with known large constant distinct input delays and large output constant delay
with a small uncertain fast-varying delay (without constraints on the delay
derivative). Note that each individual input channel may induce a different
delay (see \cite{bk17tac}), whereas the delay uncertainty may appear due to
network-based measurements (see \cite{fri14book}). We first transform the
original static quadratic map into a new one with the Hessian containing a
nominal diagonal part, and then apply a time-delay approach to the resulting
ES system to get a time-delay system, which is a perturbation of a linear
time-delay system with constant coefficients. Finally, we use the variation of
constants formula to quantitatively analyze the practical stability of the
retarded systems (and thus of the original ES systems). In the stability
analysis, we exploit positivity of the fundamental solution that corresponds
to the nominal time-delay system to obtain tighter bounds. Moreover, the
sampled-data ES in the presence of large distinct input delays is also
presented. Explicit conditions in terms of simple inequalities are established
to guarantee the practical stability of the ES control systems. Through the
solution of the constructed inequalities, we find upper bounds on the dither
period that ensures the practical stability, and also provide quantitative
ultimate bound (UB) on estimation error. We show that given any large delays
and initial box, by choosing appropriate gains we can achieve practical
stability for fast enough dithers and small enough uncertainties.
Note that if the Hessian is completely unknown, quantitative results seem to be not possible, since the Hessian
is used in ES algorithm to provide an estimate of the gradient.

We summarize the contribution as follows: 1) For the first time, quantitative
conditions are presented for ES of partly known static quadratic maps in the
presence of large input and output delays, where uncertain fast-varying
measurement delay (that may appear due to e.g. network-based measurements) and
distinct input delays are taken into account. Note that in the existing
qualitative results \cite{mk21auto,okt17tac} the case of distinct delays in
the map (and not in the inputs) is considered, whereas in \cite{orddk2018acc,rokg21ijacs} the delays
are known and slowly-varying (with the delay derivative less than one). 2)
Differently from the existing ES results via the time-delay approach
\cite{yf22tac,zf22auto,zfo21cdc,zfo22tac}, in the present paper we suggest to start
with diagonalization of the nominal part of the quadratic map and further
exploit positivity of the fundamental solutions of the time-delay equations
with their tight bounds. This leads to simplified and less conservative
conditions. 3) For the first time we present sampled-data implementation of ES
control of nD quadratic static map in the presence of large distinct input
delays. Note that in \cite{zfo21cdc,zfo22tac}, one and two input variables
with a single delay of the order of the small parameter were considered.

The paper's rest organization is as follows: In Section \ref{sec2}, we present
some notation and preliminaries. In Section \ref{sec3} and Section \ref{sec4},
we apply the time-delay approach to gradient-based classical ES and
sampled-data ES, respectively. The latter sections contain two parts: the
theoretical results and examples with simulations. Section \ref{sec5}
concludes this paper.

\section{Notation and Preliminaries\label{sec2}}

\textbf{Notation:} The notation used in this paper is fairly standard. The
notations $\mathbf{Z}_{+}$ and $\mathbf{N}$ refer to the set of nonnegative
integers and positive integers, respectively. The notation $P>0$ ($<0$) for
$P\in \mathbf{R}^{n\times n}$ means that $P$ is symmetric and positive
(negative) definite. $I_{n}$ is the $n\times n$ identity matrix. The notations
$\left \vert \cdot \right \vert $ and $\left \Vert \cdot \right \Vert $ refer to the
Euclidean vector norm and the induced matrix $2$ norm, respectively.

We will employ the variation of constants formula for delay differential
equations and some properties of the corresponding fundamental solutions as in
the following lemma, these results are brought from \cite{abbd12book} (see
Lemma 2.1, Theorem 2.7 and Corollary 2.14).

\begin{lemma}
\label{lemma1}Consider the following scalar delay differential equation:%
\begin{equation}
\left.  \dot{x}(t)+ax(t-g(t))=f(t),\text{ }t\geq t_{0}\right.  \label{eq30}%
\end{equation}
with the initial value%
\begin{equation}
x(t)=\varphi(t),\text{ }t<t_{0},\text{ }x(t_{0})=x_{0}, \label{eq31}%
\end{equation}
where $g:[t_{0},\infty)\rightarrow \mathbf{R}$ is a Lebesgue measurable
function, $f:[t_{0},\infty)\rightarrow \mathbf{R}$ is a Lebesgue measurable
locally essentially bounded function and $\varphi:(-\infty,t_{0}%
)\rightarrow \mathbf{R}$ is a piecewise continuous and bounded function.
Then there exists one and only one solution for (\ref{eq30})-(\ref{eq31}) as
in the following form%
\begin{equation}
\left.
\begin{array}
[c]{l}%
x(t)=X(t,t_{0})x_{0}+%
{\textstyle \int \nolimits_{t_{0}}^{t}}
X(t,s)f(s)\mathrm{d}s\\
\text{ \  \  \  \  \  \  \  \ }-%
{\textstyle \int \nolimits_{t_{0}}^{t}}
X(t,s)a\varphi(s-g(s))\mathrm{d}s,
\end{array}
\right.  \label{eq77}%
\end{equation}
where $\varphi(s-g(s))=0$ if $s-g(s)>t_{0}\mathbf{,}$ and the fundamental
solution $X(t,s)$ is the solution of%
\[
\left.
\begin{array}
[c]{l}%
\dot{x}(t)+ax(t-g(t))=0,\text{ }t\geq s,\\
x(t)=0,\text{ }t<s,\text{ }x(s)=1.
\end{array}
\right.
\]
Particularly, let $a>0,$ $0\leq g(t)\leq L,$ $t\geq t_{0}$ and%
\[
\left.  LA\leq \frac{1}{\mathrm{e}},\text{ }t\geq t_{0}.\right.
\]
Then for $t>s\geq t_{0},$%
\[
\left.  0<X(t,s)\leq \left \{
\begin{array}
[c]{ll}%
1, & s\leq t\leq s+L,\\
\mathrm{e}^{-a(t-s-L)}, & t\geq s+L.
\end{array}
\right.  \right.
\]

\end{lemma}

\section{ES with Large Distinct Delays\label{sec3}}

\subsection{Multi-variable static map}

Consider multi-variable static maps with measurement delay given by:
\begin{equation}
\left.
\begin{array}
[c]{l}%
y(t)=Q(\theta(t-D_{\mathrm{out}}(t)))\\
\text{ \  \  \  \ }=Q^{\ast}+\frac{1}{2}[\theta(t-D_{\mathrm{out}}(t))-\theta
^{\ast}]^{\mathrm{T}}\\
\text{ \  \  \  \  \  \  \ }\times H[\theta(t-D_{\mathrm{out}}(t))-\theta^{\ast}],
\end{array}
\right.  \label{eq1}
\end{equation}
where $y(t)\in \mathbf{R}$ is the measurable output, $\theta(t)\in
\mathbf{R}^{n}$ is the vector input, $Q^{\ast}\in \mathbf{R}$ and $\theta
^{\ast}\in \mathbf{R}^{n}$ are constants, $H=H^{\mathrm{T}}$ is the Hessian
matrix which is either positive definite or negative definite, and
$D_{\mathrm{out}}(t)=D+\Delta D(t).$ Here $D>0$ is a large and known constant
and $\Delta D(t)$ is a small time-varying uncertainty satisfying $\left \vert
\Delta D(t)\right \vert \leq \rho$, where $\rho \geq0$ is a small known constant.
The fast-varying (without any constraints on the delay derivative) delay
uncertainty may appear due to sampling of the measurement data or
network-based measurements (see \cite{fri14book}, Chapter 7). Moreover, we
also consider the ES in the presence of distinct known constant input delays
$D_{i}^{\mathrm{in}}>0$ $(i=1,\ldots,n)$.\ For future use, we denote for
$i=1,\ldots,n\mathbf{,}$
\begin{equation}
\left.
\begin{array}
[c]{l}%
D_{i}(t)=D_{i}^{\mathrm{in}}+D_{\mathrm{out}}(t),\text{ }\bar{D}_{i}%
=D_{i}^{\mathrm{in}}+D,\\
D_{i\mathrm{M}}=D_{i}^{\mathrm{in}}+D+\rho,\text{ }D_{\text{\textrm{M}}}%
=\max_{i=1,\ldots,n}D_{i\mathrm{M}}.
\end{array}
\right.  \label{eq40a}
\end{equation}
We assume that the nominal values of delays $\bar{D}_{i}$ are commensurable:

\textbf{A1.} $\bar{D}_{i}/\bar{D}_{j}$ $(\forall i,j=1,\ldots,n)$ are
rational, meaning that for some $m_{i}\in \mathbf{N}$ $(i=1,\ldots,n)$ the
following holds:
\begin{equation}
m_{i}\bar{D}_{j}=m_{j}\bar{D}_{i},\text{ }\forall i,j=1,\ldots,n\mathbf{.}
\label{eq61a}
\end{equation}

In the design, we will choose the corresponding $m_{i}$ $(i=1,\ldots,n)$ as
small as possible.

Without loss of generality, we assume that the quadratic map (\ref{eq1}) has a
minimum value $y(t)=Q^{\ast}$ at $\theta=\theta^{\ast}$, and then $H>0$. In
order to derive efficient quantitative conditions, we further assume that:

\textbf{A2. }The extremum point $\theta^{\ast}=[\theta_{1}^{\ast}%
,\ldots,\theta_{n}^{\ast}]^{\mathrm{T}}$ to be sought is uncertain from some
known box $\theta_{i}^{\ast}\in \lbrack \underline{\theta}_{i}^{\ast}%
,\bar{\theta}_{i}^{\ast}]$ with $\left \vert \bar{\theta}_{i}^{\ast}%
-\underline{\theta}_{i}^{\ast}\right \vert =\sigma_{0i},$ $i=1,\ldots,n.$

\textbf{A3.} The extremum value $Q^{\ast}$ is unknown, but it is subject to
$\left \vert Q^{\ast}\right \vert \leq Q_{\mathrm{M}}^{\ast}$ with
$Q_{\mathrm{M}}^{\ast}$ being known.

\textbf{A4.} The Hessian $H$ is uncertain and subject to $H=\bar{H}+\Delta H$
with $\bar{H}>0$ being known nominal part of $H$ and $\left \Vert \Delta
H\right \Vert \leq \kappa.$ Here $\kappa \geq0$ is a known scalar.

\begin{remark}
In classical ES, the Hessian $H$, the extremum value $Q^{\ast}$ and the
extremum point $\theta^{\ast}$ in (\ref{eq1}) are assumed to be unknown, where
tuning parameters may be found from simulations only. Here we study a "grey
box" model with Assumptions \textbf{A2-A4} and provide a quantitative analysis.
There is a tradeoff between the quantitative analysis with the plant
information and the qualitative analysis without the model knowledge.
\end{remark}

Since $\bar{H}>0$ is known, we can find an orthogonal matrix $U\in
\mathbf{R}^{n\times n}$ (obviously, $\left \Vert U\right \Vert =1$) such that
\begin{equation}
\left.  U\bar{H}U^{\mathrm{T}}=\mathrm{diag}\{ \bar{h}_{1},\ldots,\bar{h}%
_{n}\}>0.\right.  \label{eq32a}
\end{equation}
Let
\begin{equation}
\left.
\begin{array}
[c]{l}%
\mathcal{H}=UHU^{\mathrm{T}},\text{ }\mathcal{\bar{H}}=U\bar{H}U^{\mathrm{T}%
},\text{ }\Delta \mathcal{H}=U\Delta HU^{\mathrm{T}},\\
\vartheta(t)=U\theta(t),\text{ }\vartheta^{\ast}=U\theta^{\ast}.
\end{array}
\right.  \label{eq32}
\end{equation}
Then the cost function (\ref{eq1}) can be rewritten as
\begin{equation}
\left.
\begin{array}
[c]{l}%
y(t)=Q(U^{\mathrm{T}}\vartheta(t-D_{\mathrm{out}}(t)))\\
\text{ \  \  \  \ }=Q^{\ast}+\frac{1}{2}[\vartheta(t-D_{\mathrm{out}%
}(t))-\vartheta^{\ast}]^{\mathrm{T}}\\
\text{ \  \  \  \  \  \  \ }\times \mathcal{H}[\vartheta(t-D_{\mathrm{out}%
}(t))-\vartheta^{\ast}]
\end{array}
\right.  \label{eq26}
\end{equation}
with the diagonal nominal part $\mathcal{\bar{H}}$ of the Hessian
$\mathcal{H}$:
\begin{equation}
\mathcal{H}=\mathcal{\bar{H}}+\Delta \mathcal{H},\text{ }\mathcal{\bar{H}%
}=\mathrm{diag}\{ \bar{h}_{1},\bar{h}_{2},\ldots,\bar{h}_{n}\}. \label{eq11}
\end{equation}
Now define the real-time estimates $\hat{\theta}(t)$ and $\hat{\vartheta}(t)$
of $\theta^{\ast}$ and $\vartheta^{\ast}$, respectively, with the estimation
errors:
\begin{equation}
\left.  \tilde{\theta}(t)=\hat{\theta}(t)-\theta^{\ast},\text{ }%
\tilde{\vartheta}(t)=\hat{\vartheta}(t)-\vartheta^{\ast}.\right.
\label{eq55b}
\end{equation}
Design $\hat{\theta}(t)=U^{\mathrm{T}}\hat{\vartheta}(t),$ then from
(\ref{eq32}) and (\ref{eq55b}), we have $\tilde{\theta}(t)=U^{\mathrm{T}%
}\tilde{\vartheta}(t)$, which implies $|\tilde{\theta}(t)|=|\tilde{\vartheta
}(t)|.$ Therefore, it is sufficient to find bounds on $\tilde{\vartheta}(t)$.

In the ES, we use the dither signals $S(t)$ and $M(t)$ as:
\begin{equation}
\left.
\begin{array}
[c]{l}%
S(t)=[S_{1}(t),\ldots,S_{n}(t)]^{\mathrm{T}},S_{i}(t)=a_{i}\sin(\omega
_{i}t),\\
M(t)=[M_{1}(t),\ldots,M_{n}(t)]^{\mathrm{T}},M_{i}(t)=\frac{2}{a_{i}}%
\sin(\omega_{i}t),
\end{array}
\right.  \label{eq55}%
\end{equation}
where the frequencies $\omega_{i}\neq \omega_{j},i\neq j$ are non-zero,
$\omega_{i}/\omega_{j}$ are rational and $a_{i}$ are non-zero real numbers. Choose the
adaptation gain
\begin{equation}
\left.  K=\mathrm{diag}\{k_{1},\ldots,k_{n}\},\text{ }k_{i}<0\right.  \label{eq40}
\end{equation}
such that $K\mathcal{\bar{H}}$ is Hurwitz. Let the control law be in the form:
\begin{equation}
\left.
\begin{array}
[c]{l}%
u_{i}(t)=k_{i}M_{i}(t+D_{i}^{\mathrm{in}})y(t),\\
\vartheta_{i}(t)=\hat{\vartheta}_{i}(t)+S_{i}(t),\text{ }t\geq
D_{\text{\textrm{M}}}.
\end{array}
\right.  \label{eq41a}%
\end{equation}
Then the gradient-based classical ES algorithm in the presence of distinct
input delays is governed by ($i=1,\ldots
,n$)
\begin{equation}
\left.  \dot{\hat{\vartheta}}_{i}(t)=\left \{
\begin{array}
[c]{lc}%
0, & t\in \lbrack0,D_{\text{\textrm{M}}}),\\
u_{i}(t-D_{i}^{\mathrm{in}}), & t\in \lbrack D_{\text{\textrm{M}}},\infty),
\end{array}
\right.  \right.  \label{eq67}
\end{equation}
Control law (\ref{eq67}) means that we wait and start all control actions
$u_{i}(t-D_{i}^{\mathrm{in}})$ $(i=1,\ldots,n)$ at the same time
$t=D_{\text{\textrm{M}}}$, which simplifies the stability analysis.

By (\ref{eq26}), (\ref{eq55b}) and (\ref{eq55})-(\ref{eq67}), the estimation
error is governed by ($i=1,\ldots,n$)
\begin{equation}
\left.
\begin{array}
[c]{l}%
\dot{\tilde{\vartheta}}_{i}(t)=k_{i}M_{i}(t)Q^{\ast}+\frac{1}{2}k_{i}%
M_{i}(t)\\
\text{ \  \ }\times \lbrack S(t-D_{i}(t))+\tilde{\vartheta}(t-D_{i}%
(t))]^{\mathrm{T}}\mathcal{H}\\
\text{ \  \ }\times \lbrack S(t-D_{i}(t))+\tilde{\vartheta}(t-D_{i}(t))],\text{
}t\geq D_{\text{\textrm{M}}}\mathbf{,}\\
\tilde{\vartheta}_{i}(t)=\tilde{\vartheta}_{i}(D_{\text{\textrm{M}}}),\text{
}t\in \lbrack0,D_{\text{\textrm{M}}}].
\end{array}
\right.  \label{eq35}
\end{equation}

\begin{remark}
\label{remark1}It is seen from (\ref{eq35}) that in each subsystem, the delay
$D_{i}(t)$ keeps the same value for the whole state $\tilde{\vartheta}$, which
is different from \cite{mk21auto,okt17tac} with a simplified form
$\tilde{\vartheta}\left(  t-D\right)  :=[\tilde{\vartheta}_{1}\left(
t-D_{1}\right)  ,\ldots,\tilde{\vartheta}_{n}\left(  t-D_{n}\right)
]^{\mathrm{T}}.$ The latter form corresponds to an abstract map (\ref{eq26})
with $\vartheta(t-D_{\mathrm{out}%
}(t))$ changed by $\vartheta(t-D_{\mathrm{out}%
}(t))=[\vartheta_{1}(t-D_{1}%
),\ldots,\vartheta_{n}(t-D_{n})]^{\mathrm{T}}$ studied in these papers, but
not to the case of distinct input delays in the ES algorithm. Note that in
this simplified case, one do not need assumption \textbf{A4} on commensurable
delays $\bar{D}_{i}$ as well as on the special choice of dither frequencies
(as defined by (\ref{eq66a}) of the following Lemma \ref{lemma2}). The
time-varying delays case considered here is also more general than the
constant delay case in \cite{mk21auto,okt17tac}.
\end{remark}

\begin{remark}
Denote $U=[b_{ij}]_{n\times n}$. Since $\theta(t)=U^{\mathrm{T}}\vartheta(t)$
and $\hat{\theta}(t)=U^{\mathrm{T}}\hat{\vartheta}(t),$ then from
(\ref{eq41a}) and (\ref{eq67}) it is easy to present the gradient-based
classical ES algorithm for the original ES problem as follows ($i=1,\ldots
,n$):
\[
\left.  \dot{\hat{\theta}}_{i}(t)=\left \{
\begin{array}
[c]{lc}%
0, & t\in \lbrack0,D_{\text{\textrm{M}}}),\\%
{\textstyle \sum \limits_{j=1}^{n}}
b_{ji}u_{j}(t-D_{j}^{\mathrm{in}}), & t\in \lbrack D_{\text{\textrm{M}}}%
,\infty),
\end{array}
\right.  \right.
\]
with
\[
\left.
\begin{array}
[c]{l}%
u_{i}(t)=k_{i}M_{i}(t+D_{i}^{\mathrm{in}})y(t),\\
\theta_{i}(t)=\hat{\theta}_{i}(t)+%
{\textstyle \sum \limits_{j=1}^{n}}
b_{ji}S_{j}(t),\text{ }t\geq D_{\text{\textrm{M}}}.
\end{array}
\right.
\]
By using this algorithm and noting that $\tilde{\theta}(t)=U^{\mathrm{T}}%
\tilde{\vartheta}(t),$ we can obtain part (ii) in Theorem 1 below. The
corresponding analysis is also suitable for the sampled-data ES case in Section 4.
\end{remark}

To choose appropriate dither frequencies, we will employ the following lemma,
which is proved in Appendix:

\begin{lemma}
\label{lemma2}Under \textbf{A1, }consider the following positive numbers
(frequencies):
\begin{equation}
\left.  \omega_{i}=\frac{q\cdot2\pi im_{j}}{\bar{D}_{j}},\text{ }%
q\in \mathbf{N},\text{ }\forall i,j=1,\ldots,n\mathbf{.}\right.  \label{eq66a}
\end{equation}
Then the following relations hold:
\begin{equation}
\left.  \sin(\omega_{i}(t-\bar{D}_{j}))=\sin(\omega_{i}t),\text{ }\forall
i,j=1,\ldots,n\mathbf{.}\right.  \label{eq66}%
\end{equation}
\end{lemma}

Based on Lemma \ref{lemma2}, we assume that:

\textbf{A5.} The frequencies are chosen according to (\ref{eq66a}) with some
$q\in \mathbf{N}$.

Due to (\ref{eq66a}), the frequencies $\omega_{i}=\omega_{i,q}$ $(i=1,\ldots
,n)$ can be always chosen larger than any $\omega^{\ast}>0$. Using Lemma
\ref{lemma2}, we have $\sin(\omega_{i}(t-D_{j}(t)))=\sin(\omega_{i}(t-\Delta
D(t)))$ ($i,j=1,\ldots,n$\textbf{),} which implies that $S(t-D_{i}%
(t))=S(t-\Delta D(t))$ ($i=1,\ldots,n$\textbf{)}. Then system (\ref{eq35}) can
be rewritten as
\begin{equation}
\left.
\begin{array}
[c]{l}%
\dot{\tilde{\vartheta}}_{i}(t)=k_{i}M_{i}(t)Q^{\ast}+\frac{1}{2}k_{i}%
M_{i}(t)\\
\text{ \  \ }\times \lbrack S(t-\Delta D(t))+\tilde{\vartheta}(t-D_{i}%
(t))]^{\mathrm{T}}\mathcal{H}\\
\text{ \  \ }\times \lbrack S(t-\Delta D(t))+\tilde{\vartheta}(t-D_{i}%
(t))],\text{ }t\geq D_{\text{\textrm{M}}},\\
\tilde{\vartheta}_{i}(t)=\tilde{\vartheta}_{i}(D_{\text{\textrm{M}}}),\text{
}t\in \lbrack0,D_{\text{\textrm{M}}}].
\end{array}
\right.  \label{eq35b}%
\end{equation}
Taking into account $S(t-\Delta D(t))=S(t)-%
{\textstyle \int \nolimits_{t-\Delta D(t)}^{t}}
\dot{S}(s)\mathrm{d}s,$ system (\ref{eq35b}) can be further expressed as
\begin{equation}
\left.
\begin{array}
[c]{l}%
\dot{\tilde{\vartheta}}_{i}(t)=k_{i}M_{i}(t)Q^{\ast}+\frac{1}{2}k_{i}%
M_{i}(t)S^{\mathrm{T}}(t)\mathcal{H}S(t)\\
+\frac{1}{2}k_{i}M_{i}(t)\tilde{\vartheta}^{\mathrm{T}}(t-D_{i}(t))\mathcal{H}%
\tilde{\vartheta}(t-D_{i}(t))\\
+k_{i}M_{i}(t)S^{\mathrm{T}}(t)\mathcal{H}\tilde{\vartheta}(t-D_{i}%
(t))+w_{i}(t),\text{ }t\geq D_{\text{\textrm{M}}},\\
\tilde{\vartheta}_{i}(t)=\tilde{\vartheta}_{i}(D_{\text{\textrm{M}}}),\text{
}t\in \lbrack0,D_{\text{\textrm{M}}}]
\end{array}
\right.  \label{eq112}
\end{equation}
with
\begin{equation}
\left.
\begin{array}
[c]{l}%
w_{i}(t)=-k_{i}M_{i}(t)S^{\mathrm{T}}(t)\mathcal{H}%
{\textstyle \int \nolimits_{t-\Delta D(t)}^{t}}
\dot{S}(s)\mathrm{d}s\\
+\frac{1}{2}k_{i}M_{i}(t)\left(
{\textstyle \int \nolimits_{t-\Delta D(t)}^{t}}
\dot{S}(s)\mathrm{d}s\right)  ^{\mathrm{T}}\mathcal{H}\left(
{\textstyle \int \nolimits_{t-\Delta D(t)}^{t}}
\dot{S}(s)\mathrm{d}s\right)  \\
-k_{i}M_{i}(t)\left(
{\textstyle \int \nolimits_{t-\Delta D(t)}^{t}}
\dot{S}(s)\mathrm{d}s\right)  ^{\mathrm{T}}\mathcal{H}\tilde{\vartheta
}(t-D_{i}(t)).
\end{array}
\right.  \label{eq33a}
\end{equation}
Due to
\begin{equation}
\left.  \frac{\bar{D}_{1}}{qm_{1}}=\cdots=\frac{\bar{D}_{n}}{qm_{n}%
}:=\varepsilon,\text{ }q\in \mathbf{N},\right.  \label{eq34}
\end{equation}
from (\ref{eq66a}) we have
\begin{equation}
\left.  \omega_{i}=\frac{2\pi i}{\varepsilon},\text{ }i=1,\ldots
,n\mathbf{.}\right.  \label{eq34a}
\end{equation}

For the stability analysis of the ES control system (\ref{eq112}), inspired by
\cite{fz20auto,zf22auto}, we first apply the time-delay approach to averaging
of (\ref{eq112}). Integrating from $t-\varepsilon$ to $t$ and dividing by
$\varepsilon$ on both sides of equation (\ref{eq112}), we get
\begin{equation}
\left.
\begin{array}
[c]{l}%
\frac{1}{\varepsilon}%
{\textstyle \int \nolimits_{t-\varepsilon}^{t}}
\dot{\tilde{\vartheta}}_{i}(\tau)\mathrm{d}\tau=\frac{1}{\varepsilon}%
{\textstyle \int \nolimits_{t-\varepsilon}^{t}}
k_{i}M_{i}(\tau)Q^{\ast}\mathrm{d}\tau \\
+\frac{1}{2\varepsilon}%
{\textstyle \int \nolimits_{t-\varepsilon}^{t}}
k_{i}M_{i}(\tau)S^{\mathrm{T}}(\tau)\mathcal{H}S(\tau)\mathrm{d}\tau \\
+\frac{1}{2\varepsilon}%
{\textstyle \int \nolimits_{t-\varepsilon}^{t}}
k_{i}M_{i}(\tau)\tilde{\vartheta}^{\mathrm{T}}(\tau-D_{i}(\tau))\mathcal{H}%
\tilde{\vartheta}(\tau-D_{i}(\tau))\mathrm{d}\tau \\
+\frac{1}{\varepsilon}%
{\textstyle \int \nolimits_{t-\varepsilon}^{t}}
k_{i}M_{i}(\tau)S^{\mathrm{T}}(\tau)\mathcal{H}\tilde{\vartheta}(\tau
-D_{i}(\tau))\mathrm{d}\tau \\
+\frac{1}{\varepsilon}%
{\textstyle \int \nolimits_{t-\varepsilon}^{t}}
w_{i}(\tau)\mathrm{d}\tau,\text{ }t\geq D_{\text{\textrm{M}}}+\varepsilon
,\text{ }i=1,\ldots,n.
\end{array}
\right.  \label{eq12}
\end{equation}
Note for $\forall i,j,k=1,\ldots,n\mathbf{,}$ there hold
\[
\left.
\begin{array}
[c]{l}%
\int \nolimits_{t-\varepsilon}^{t}\sin \left(  \frac{2\pi i}{\varepsilon}%
\tau \right)  \mathrm{d}\tau=0,\\
\int \nolimits_{t-\varepsilon}^{t}\sin \left(  \frac{2\pi i}{\varepsilon}%
\tau \right)  \sin \left(  \frac{2\pi j}{\varepsilon}\tau \right)  \sin \left(
\frac{2\pi k}{\varepsilon}\tau \right)  \mathrm{d}\tau=0,
\end{array}
\right.
\]
then for $i=1,\ldots,n\mathbf{,}$
\begin{equation}
\left.
\begin{array}
[c]{l}%
\frac{1}{\varepsilon}%
{\textstyle \int \nolimits_{t-\varepsilon}^{t}}
k_{i}M_{i}(\tau)Q^{\ast}\mathrm{d}\tau=0,\\
\frac{1}{2\varepsilon}%
{\textstyle \int \nolimits_{t-\varepsilon}^{t}}
k_{i}M_{i}(\tau)S^{\mathrm{T}}(\tau)\mathcal{H}S(\tau)\mathrm{d}%
\tau=0\mathbf{.}%
\end{array}
\right.  \label{eq48}
\end{equation}
For the third term on the right-hand side of (\ref{eq12}), we have
\begin{equation}
\left.
\begin{array}
[c]{l}%
\text{ \  \ }\frac{1}{2\varepsilon}%
{\textstyle \int \nolimits_{t-\varepsilon}^{t}}
k_{i}M_{i}(\tau)\tilde{\vartheta}^{\mathrm{T}}(\tau-D_{i}(\tau))\mathcal{H}%
\tilde{\vartheta}(\tau-D_{i}(\tau))\mathrm{d}\tau \\
=\frac{1}{2\varepsilon}%
{\textstyle \int \nolimits_{t-\varepsilon}^{t}}
k_{i}M_{i}(\tau)\mathrm{d}\tau \tilde{\vartheta}^{\mathrm{T}}(t-D_{i}%
(t))\mathcal{H}\tilde{\vartheta}(t-D_{i}(t))\\
\text{ \ }-\frac{1}{2\varepsilon}%
{\textstyle \int \nolimits_{t-\varepsilon}^{t}}
k_{i}M_{i}(\tau)[\tilde{\vartheta}^{\mathrm{T}}(t-D_{i}(t))\mathcal{H}%
\tilde{\vartheta}(t-D_{i}(t))\\
\text{ \ }-\tilde{\vartheta}^{\mathrm{T}}(\tau-D_{i}(\tau))\mathcal{H}%
\tilde{\vartheta}(\tau-D_{i}(\tau))]\mathrm{d}\tau \\
=-\frac{1}{\varepsilon}%
{\textstyle \int \nolimits_{t-\varepsilon}^{t}}
{\textstyle \int \nolimits_{\tau-D_{i}(\tau)}^{t-D_{i}(t)}}
k_{i}M_{i}(\tau)\tilde{\vartheta}^{\mathrm{T}}(s)\mathcal{H}\dot
{\tilde{\vartheta}}(s)\mathrm{d}s\mathrm{d}\tau,
\end{array}
\right.  \label{eq50}
\end{equation}
where we have used $%
{\textstyle \int \nolimits_{t-\varepsilon}^{t}}
k_{i}M_{i}(\tau)\mathrm{d}\tau=0.$ For the fourth term on the right-hand side
of (\ref{eq12}), we obtain
\begin{equation}
\left.
\begin{array}
[c]{l}%
\text{ \  \ }\frac{1}{\varepsilon}%
{\textstyle \int \nolimits_{t-\varepsilon}^{t}}
k_{i}M_{i}(\tau)S^{\mathrm{T}}(\tau)\mathcal{H}\tilde{\vartheta}(\tau
-D_{i}(\tau))\mathrm{d}\tau \\
=\frac{1}{\varepsilon}k_{i}%
{\textstyle \int \nolimits_{t-\varepsilon}^{t}}
M_{i}(\tau)S^{\mathrm{T}}(\tau)\mathrm{d}\tau \cdot \mathcal{H}\tilde{\vartheta
}(t-D_{i}(t))\\
\text{ \ }-\frac{1}{\varepsilon}%
{\textstyle \int \nolimits_{t-\varepsilon}^{t}}
k_{i}M_{i}(\tau)S^{\mathrm{T}}(\tau)\mathcal{H}\\
\text{ \ }\times \lbrack \tilde{\vartheta}(t-D_{i}(t))-\tilde{\vartheta}%
(\tau-D_{i}(\tau))]\mathrm{d}\tau \\
=k_{i}\bar{h}_{i}\tilde{\vartheta}_{i}(t-D_{i}(t))+C(k_{i})\Delta
\mathcal{H}\tilde{\vartheta}(t-D_{i}(t))\\
\text{ \ }-\frac{1}{\varepsilon}%
{\textstyle \int \nolimits_{t-\varepsilon}^{t}}
{\textstyle \int \nolimits_{\tau-D_{i}(\tau)}^{t-D_{i}(t)}}
k_{i}M_{i}(\tau)S^{\mathrm{T}}(\tau)\mathcal{H}\dot{\tilde{\vartheta}%
}(s)\mathrm{d}s\mathrm{d}\tau
\end{array}
\right.  \label{eq51}
\end{equation}
with $C(k_{i})=\left[  0,\ldots,0,k_{i},0,\ldots,0\right]  ,$ where we have
noted that $\mathcal{H}=\mathcal{\bar{H}}+\Delta \mathcal{H}$ with
$\mathcal{\bar{H}=}\mathrm{diag}\{ \bar{h}_{1},\ldots,\bar{h}_{n}\}$ and
\[
\left.  \int \nolimits_{t-\varepsilon}^{t}\frac{2a_{i}}{a_{j}}\sin \left(
\frac{2\pi i}{\varepsilon}\tau \right)  \sin \left(  \frac{2\pi j}{\varepsilon
}\tau \right)  \mathrm{d}\tau=\left \{
\begin{array}
[c]{cc}%
\varepsilon, & i=j,\\
0, & i\neq j.
\end{array}
\right.  \right.
\]
Let
\begin{equation}
\left.
\begin{array}
[c]{l}%
G_{i}(t)=\frac{1}{\varepsilon}\int \nolimits_{t-\varepsilon}^{t}(\tau
-t+\varepsilon)k_{i}M_{i}(\tau)\\
\text{ \  \  \  \  \  \  \  \  \  \ }\times \bar{Q}(\tilde{\vartheta}(\tau-D_{i}%
(\tau)))\mathrm{d}\tau,\text{ }t\geq D_{\mathrm{M}}+\varepsilon
\end{array}
\right.  \label{eq20}%
\end{equation}
with
\begin{equation}
\left.
\begin{array}
[c]{l}%
\bar{Q}(\tilde{\vartheta}(\tau-D_{i}(\tau)))=Q^{\ast}+\frac{1}{2}%
[\tilde{\vartheta}(\tau-D_{i}(\tau))+S(\tau)]^{\mathrm{T}}\\
\text{ \  \  \  \  \  \  \  \  \  \  \  \  \  \  \  \  \  \  \  \  \  \  \  \  \  \ }\times
\mathcal{H}[\tilde{\vartheta}(\tau-D_{i}(\tau))+S(\tau)].
\end{array}
\right.  \label{eq20c}%
\end{equation}
Then similar to \cite{zhf22auto}, we can present
\begin{equation}
\left.
\begin{array}
[c]{c}%
\frac{1}{\varepsilon}%
{\textstyle \int \nolimits_{t-\varepsilon}^{t}}
\dot{\tilde{\vartheta}}_{i}(\tau)\mathrm{d}\tau=\frac{\mathrm{d}}{\mathrm{d}%
t}[\tilde{\vartheta}_{i}(t)-G_{i}(t)]\\
-w_{i}(t)+\frac{1}{\varepsilon}\int \nolimits_{t-\varepsilon}^{t}w_{i}%
(\tau)\mathrm{d}\tau.
\end{array}
\right.  \label{eq52}
\end{equation}
Now we set
\begin{equation}
\left.
\begin{array}
[c]{l}%
Y_{1i}(t)=\frac{1}{\varepsilon}%
{\textstyle \int \nolimits_{t-\varepsilon}^{t}}
{\textstyle \int \nolimits_{\tau-D_{i}(\tau)}^{t-D_{i}(t)}}
k_{i}M_{i}(\tau)\tilde{\vartheta}^{\mathrm{T}}(s)\mathcal{H}\dot
{\tilde{\vartheta}}(s)\mathrm{d}s\mathrm{d}\tau,\\
Y_{2i}(t)=\frac{1}{\varepsilon}%
{\textstyle \int \nolimits_{t-\varepsilon}^{t}}
{\textstyle \int \nolimits_{\tau-D_{i}(\tau)}^{t-D_{i}(t)}}
k_{i}M_{i}(\tau)S^{\mathrm{T}}(\tau)\mathcal{H}\dot{\tilde{\vartheta}%
}(s)\mathrm{d}s\mathrm{d}\tau,\\
Y_{3i}(t)=C(k_{i})\Delta \mathcal{H}\tilde{\vartheta}(t-D_{i}(t)),\text{ }t\geq
D_{\text{\textrm{M}}}+\varepsilon,
\end{array}
\right.  \label{eq20a}
\end{equation}
and denote
\begin{equation}
\left.
\begin{array}
[c]{l}%
z_{i}(t)=\tilde{\vartheta}_{i}(t)-G_{i}(t),\text{ }t\geq \varepsilon,\\
G_{i}(t)=0,\text{ }t\in \lbrack \varepsilon,D_{\text{\textrm{M}}}+\varepsilon).
\end{array}
\right.  \label{eq52a}
\end{equation}
Then employing (\ref{eq48})-(\ref{eq51}) and (\ref{eq52})-(\ref{eq52a}), we
finally transform system (\ref{eq12}) into
\begin{equation}
\left.
\begin{array}
[c]{l}%
\dot{z}_{i}(t)=k_{i}\bar{h}_{i}z_{i}(t-D_{i}(t))\\
\text{ \  \  \  \  \  \  \  \ }+\bar{w}_{i}(t),\text{ }t\geq D_{\text{\textrm{M}}%
}+\varepsilon,\text{ }i=1,\ldots,n,\\
\bar{w}_{i}(t)=k_{i}\bar{h}_{i}G_{i}(t-D_{i}(t))-Y_{1i}(t)\\
\text{ \  \  \  \  \  \  \  \ }-Y_{2i}(t)+Y_{3i}(t)+w_{i}(t).
\end{array}
\right.  \label{eq21}
\end{equation}
For the stability analysis, we choose
\begin{equation}
\left.  \rho=\mu \varepsilon \right.  \label{eq21a}
\end{equation}
with $\mu \geq0$ being a small constant. Then $D_{i\mathrm{M}}$ and
$D_{\text{\textrm{M}}}$ defined in (\ref{eq40a}) can be rewritten as
\begin{equation}
\left.
\begin{array}
[c]{l}%
D_{i\mathrm{M}}=D_{i\mathrm{M}}(\varepsilon)=D_{i}^{\mathrm{in}}%
+D+\mu \varepsilon,\text{ }i=1,\ldots,n\mathbf{,}\\
D_{\text{\textrm{M}}}=D_{\text{\textrm{M}}}(\varepsilon)=\max_{i=1,\ldots
,n}D_{i}^{\mathrm{in}}+D+\mu \varepsilon.
\end{array}
\right.  \label{eq21b}
\end{equation}
Note that if $\tilde{\vartheta}_{i}(t)$ and $\dot{\tilde{\vartheta}}_{i}(t)$
(and thus $z_{i}(t)$) are of the order of \textrm{O}$(1),$ then the terms
$G_{i}(t)$ and $Y_{1i}(t),$ $Y_{2i}(t)$ defined by (\ref{eq20}) and
(\ref{eq20a}) are of the order of \textrm{O}$(\varepsilon),$ the term
$Y_{3i}(t)$ defined in (\ref{eq20a}) is of the order of \textrm{O}$(\kappa)$
and the term $w_{i}(t)$ defined by (\ref{eq33a}) is of the order of
\textrm{O}$(\mu).$ Therefore, $\bar{w}_{i}(t)$ in (\ref{eq21}) is of the order
of \textrm{O}$(\max \{ \varepsilon,\mu,\kappa \}).$ Similar to our previous work
\cite{yf22tac}, we will analyze (\ref{eq21}) as linear system w.r.t. $z_{i}$
(but in the present paper it is delayed system) with delayed disturbance-like
\textrm{O}$(\max \{ \varepsilon,\mu,\kappa \})$-term $\bar{w}_{i}(t)$ that
depends on the solutions of (\ref{eq112}). The resulting bound on $|z_{i}|$
will lead to the bound on $\tilde{\vartheta}_{i}:|\tilde{\vartheta}_{i}%
|\leq|z_{i}|+|G_{i}|.$ The bound on $z_{i}$ will be found by utilizing
solution representation formula in Lemma \ref{lemma1}.

We will find $k_{i}$ $(i=1,\ldots,n)$ from the inequalities
\begin{equation}
\left.  \left \vert k_{i}\bar{h}_{i}\right \vert \bar{D}_{i}-\frac{1}%
{\mathrm{e}}<0,\text{ }i=1,\ldots,n\right.  \label{eq23a}
\end{equation}
with $\bar{D}_{i}$ defined in (\ref{eq40a}), which guarantee the exponential
stability with a decay rate
\begin{equation}
\delta_{i}=\left \vert k_{i}\bar{h}_{i}\right \vert <\frac{1}{\mathrm{e}\bar
{D}_{i}}\label{eq23b}
\end{equation}
of the averaged system with $\rho=0:$
\begin{equation}
\left.  \dot{z}_{i}(t)=k_{i}\bar{h}_{i}z_{i}(t-\bar{D}_{i}),\text{ }%
i=1,\ldots,n.\right.  \label{eq23c}
\end{equation}
To formulate our main result we will use the following notations:
\begin{equation}
\left.
\begin{array}
[c]{l}%
 \Delta_{1}=Q_{\mathrm{M}}^{\ast}+%
{\textstyle \sum \limits_{j=1}^{n}}
\frac{\bar{h}_{j}}{2}\left(  \bar{\sigma}_{j}+\left \vert a_{j}\right \vert
\right)  ^{2}+\frac{\kappa}{2}\left(  \sqrt{%
{\textstyle \sum \limits_{j=1}^{n}}
\bar{\sigma}_{j}^{2}}+\sqrt{%
{\textstyle \sum \limits_{j=1}^{n}}
a_{j}^{2}}\right)  ^{2},\\
\Delta_{2}=4\pi%
{\textstyle \sum \limits_{j=1}^{n}}
\left \vert j\bar{h}_{j}a_{j}\right \vert \left(  \left \vert a_{j}\right \vert
+\pi \mu \left \vert ja_{j}\right \vert +\bar{\sigma}_{j}\right)  \\
+4\pi \kappa \sqrt{%
{\textstyle \sum \limits_{j=1}^{n}}
j^{2}a_{j}^{2}}\left(  \sqrt{%
{\textstyle \sum \limits_{j=1}^{n}}
a_{j}^{2}}+\pi \mu \sqrt{%
{\textstyle \sum \limits_{j=1}^{n}}
j^{2}a_{j}^{2}}+\sqrt{%
{\textstyle \sum \limits_{j=1}^{n}}
\bar{\sigma}_{j}^{2}}\right)  ,\\
\bar{\Delta}_{1}=2\left(  1+4\mu \right)  \Delta_{1}\left(
{\textstyle \sum \limits_{j=1}^{n}}
\frac{\left \vert \bar{h}_{j}k_{j}\right \vert \bar{\sigma}_{j}}{\left \vert
a_{j}\right \vert }+\kappa \sqrt{%
{\textstyle \sum \limits_{j=1}^{n}}
\bar{\sigma}_{j}^{2}}\sqrt{%
{\textstyle \sum \limits_{j=1}^{n}}
\frac{k_{j}^{2}}{a_{j}^{2}}}\right)  ,\\
\bar{\Delta}_{2}=2\left(  1+4\mu \right)  \Delta_{1}\left(
{\textstyle \sum \limits_{j=1}^{n}}
\left \vert \bar{h}_{j}k_{j}\right \vert +\kappa \sqrt{%
{\textstyle \sum \limits_{j=1}^{n}}
\bar{\sigma}_{j}^{2}}\sqrt{%
{\textstyle \sum \limits_{j=1}^{n}}
\frac{k_{j}^{2}}{a_{j}^{2}}}\right)  .
\end{array}
\right.  \label{eq38}
\end{equation}

\begin{theorem}
\label{theorem1}Assume that \textbf{A1-A5} hold. Let $k_{i}$ $(i=1,\ldots,n)$
satisfy (\ref{eq23a}) and $D_{\mathrm{M}},D_{i\mathrm{M}}$ $(i=1,\ldots,n)$ be
given by (\ref{eq21b}). Given tuning parameters $a_{i}$ $(i=1,\ldots,n)$ and
$\mu,$ $\kappa \geq0$ as well as $\bar{\sigma}_{i}>\bar{\sigma}_{0i}>0$
$(i=1,\ldots,n),$ let there exists $\varepsilon^{\ast}>0$ that satisfy
\begin{equation}
\left.
\begin{array}
[c]{l}%
\Phi_{1}^{i}=\left \vert k_{i}\bar{h}_{i}\right \vert D_{i\mathrm{M}%
}(\varepsilon^{\ast})-\frac{1}{\mathrm{e}}\leq0,\text{ }i=1,\ldots
,n\mathbf{,}\\
\Phi_{2}^{i}=\mathrm{e}^{\left \vert k_{i}\bar{h}_{i}\right \vert D_{i\mathrm{M}%
}(\varepsilon^{\ast})}\left[  \bar{\sigma}_{0i}+\frac{3\varepsilon^{\ast
}\left \vert k_{i}\right \vert \Delta_{1}}{\left \vert a_{i}\right \vert }%
+\frac{W_{i}(\varepsilon^{\ast},\mu,\kappa)}{\left \vert k_{i}\bar{h}%
_{i}\right \vert }\right] \\
\text{ \  \  \  \  \  \ }+\frac{\varepsilon^{\ast}\left \vert k_{i}\right \vert
\Delta_{1}}{\left \vert a_{i}\right \vert }-\bar{\sigma}_{i}<0,\text{
}i=1,\ldots,n\mathbf{,}%
\end{array}
\right.  \label{eq24a}
\end{equation}
where
\begin{equation}
\left.
\begin{array}
[c]{l}%
W_{i}(\varepsilon,\mu,\kappa)=\frac{\varepsilon \left \vert k_{i}\right \vert
}{\left \vert a_{i}\right \vert }\left(  \left \vert k_{i}\bar{h}_{i}\right \vert
\Delta_{1}+\bar{\Delta}_{1}+\bar{\Delta}_{2}\right) \\
\text{ \  \  \  \  \  \  \  \  \  \  \  \ }+\frac{\mu \left \vert k_{i}\right \vert
}{\left \vert a_{i}\right \vert }\Delta_{2}+\kappa \left \vert k_{i}\right \vert
\sqrt{%
{\textstyle \sum \limits_{j=1}^{n}}
\bar{\sigma}_{j}^{2}}%
\end{array}
\right.  \label{eq24b}
\end{equation}
with $\{ \Delta_{1},\Delta_{2},\bar{\Delta}_{1},\bar{\Delta}_{2}\}$ given by
(\ref{eq38}). Then for all $\varepsilon \in(0,\varepsilon^{\ast}]$ satisfying
(\ref{eq34}), the following holds:

\textbf{(i) }Solutions of (\ref{eq112}) with $|\tilde{\vartheta}%
_{i}(D_{\mathrm{M}})|\leq \bar{\sigma}_{0i}$ satisfy the following bounds:
\begin{equation}
\left.
\begin{array}
[c]{l}%
|\tilde{\vartheta}_{i}(t)|< \vert \tilde{\vartheta}_{i}(D_{\mathrm{M}})
\vert+\frac{2\varepsilon \left \vert k_{i}\right \vert \Delta_{1}}{\left \vert
a_{i}\right \vert }<\bar{\sigma}_{i},\text{ }t\in \lbrack D_{\mathrm{M}%
},D_{\mathrm{M}}+\varepsilon],\\
|\tilde{\vartheta}_{i}(t)|<\left(  1+\left \vert k_{i}\bar{h}_{i}\right \vert
D_{i\mathrm{M}}\right)  \left[  \vert \tilde{\vartheta}_{i}(D_{\mathrm{M}})
\vert+\frac{3\varepsilon \left \vert k_{i}\right \vert \Delta_{1}}{\left \vert
a_{i}\right \vert }\right] \\
\text{ \  \  \  \  \  \  \  \  \  \  \ }+D_{i\mathrm{M}}W_{i}(\varepsilon,\mu
,\kappa)+\frac{\varepsilon \left \vert k_{i}\right \vert \Delta_{1}}{\left \vert
a_{i}\right \vert }\\
\text{ \  \  \  \  \  \  \  \ }<\bar{\sigma}_{i},\text{ }t\in \lbrack D_{\mathrm{M}%
}+\varepsilon,D_{\mathrm{M}}+D_{i\mathrm{M}}+\varepsilon],\\
|\tilde{\vartheta}_{i}(t)|<\mathrm{e}^{-\left \vert k_{i}\bar{h}_{i}\right \vert
\left(  t-D_{\mathrm{M}}-2D_{i\mathrm{M}}-\varepsilon \right)  }\left[
\vert \tilde{\vartheta}_{i}(D_{\mathrm{M}}) \vert+\frac{3\varepsilon \left \vert
k_{i}\right \vert \Delta_{1}}{\left \vert a_{i}\right \vert }\right] \\
\text{ \  \  \  \  \  \  \  \  \  \  \ }+\frac{\mathrm{e}^{\left \vert k_{i}\bar{h}%
_{i}\right \vert D_{i\mathrm{M}}}}{\left \vert k_{i}\bar{h}_{i}\right \vert
}W_{i}(\varepsilon,\mu,\kappa)+\frac{\varepsilon \left \vert k_{i}\right \vert
\Delta_{1}}{\left \vert a_{i}\right \vert }\\
\text{ \  \  \  \  \  \  \  \ }<\bar{\sigma}_{i},\text{ }t\geq D_{\mathrm{M}%
}+D_{i\mathrm{M}}+\varepsilon.
\end{array}
\right.  \label{eq65a}
\end{equation}
These solutions are exponentially attracted to the box
\begin{equation}
\left.
\begin{array}
[c]{c}%
\vert \tilde{\vartheta}_{i}(t) \vert<\Omega_{i}\triangleq \frac{\mathrm{e}%
^{\left \vert k_{i}\bar{h}_{i}\right \vert D_{i\mathrm{M}}}W_{i}(\varepsilon
,\mu,\kappa)}{\left \vert k_{i}\bar{h}_{i}\right \vert }+\frac{\varepsilon
\left \vert k_{i}\right \vert \Delta_{1}}{\left \vert a_{i}\right \vert },\\
i=1,\ldots,n
\end{array}
\right.  \label{eq74}
\end{equation}
with decay rates $\delta_{i}=\left \vert k_{i}\bar{h}_{i}\right \vert ,$ which
are independent of $\varepsilon$ and delays $D_{i}(t).$

\textbf{(ii)} Consider $\hat{\theta}(t)=U^{\mathrm{T}}\hat{\vartheta}(t)$ with
$U=[b_{ij}]_{n\times n},$ where $\hat{\vartheta}(t)$ is defined by
(\ref{eq67}). Then the estimation errors $\tilde{\theta}_{i}(t)=\hat{\theta
}_{i}(t)-\theta^{\ast}$ such that $|\tilde{\theta}_{i}(D_{\mathrm{M}}%
)|\leq \sigma_{0i}$ with $%
{\textstyle \sum \nolimits_{i=1}^{n}}
\left \vert b_{ji}\right \vert \sigma_{0i}=\bar{\sigma}_{0j}$ $(j=1,\ldots
,n\mathbf{)}$ are exponentially attracted to the box $\vert \tilde{\theta}%
_{i}(t) \vert<\sum_{j=1}^{n}\left \vert b_{ji}\right \vert \Omega_{j}$
$(i=1,\ldots,n)$\ with a decay rate $\delta=\min_{i=1,\ldots,n}\delta_{i}%
=\min_{i=1,\ldots,n}\left \vert k_{i}\bar{h}_{i}\right \vert ,$ where
$\Omega_{j}$ is defined by (\ref{eq74}).
\end{theorem}

\begin{proof}
See Appendix A2.
\end{proof}

\begin{remark}
\label{remark0}In the case of the single input delay $D_{i}^{\mathrm{in}%
}=D^{\mathrm{in}}>0$ $(i=1,\ldots,n)$, by choosing the dither signal
$S(t)=\left[  S_{1}(t+D^{\mathrm{in}}+D),\ldots,S_{n}(t+D^{\mathrm{in}%
}+D)\right]  ^{\mathrm{T}}$ with $S_{i}(t)$ given in (\ref{eq55}), the result
of Theorem \ref{theorem1} holds for all $\varepsilon \in(0,\varepsilon^{\ast}%
]$. Indeed, in this case we can directly obtain system (\ref{eq35b}) without
using Lemma \ref{lemma2}.
\end{remark}

\begin{remark}
\label{remark2}(regional finite-time practical stability) For the fixed
parameters $\{ \varepsilon,\mu,\kappa,a_{i},k_{i},\bar{\sigma}_{0i},\bar
{\sigma}_{i}\}$ $(i=1,\ldots,n)$ and given $\Delta \Omega_{i}>0$ $(i=1,\ldots
,n)$, we show that there exists a finite time $T_{i}>0$ such that
$|\tilde{\vartheta}_{i}(t)|<\Omega_{i}+\Delta \Omega_{i},$ $\forall t\geq
T_{i}.$ Denote $\Xi_{i}=\bar{\sigma}_{0i}+\frac{3\varepsilon \left \vert
k_{i}\right \vert \Delta_{1}}{\left \vert a_{i}\right \vert }.$ Then from
(\ref{eq65a}) we find $|\tilde{\vartheta}_{i}(t)|<\mathrm{e}^{-\left \vert
k_{i}\bar{h}_{i}\right \vert \left(  t-D_{\mathrm{M}}-2D_{i\mathrm{M}%
}-\varepsilon \right)  }\Xi_{i}+\Omega_{i},$ $t\geq0.$ Let $\mathrm{e}%
^{-\left \vert k_{i}\bar{h}_{i}\right \vert \left(  t-D_{\mathrm{M}%
}-2D_{i\mathrm{M}}-\varepsilon \right)  }\Xi_{i}+\Omega_{i}\leq \Omega
_{i}+\Delta \Omega_{i}.$ Then we can obtain $t\geq T_{i}=D_{\mathrm{M}%
}+2D_{i\mathrm{M}}+\varepsilon-\ln(\Delta \Omega_{i}/\Xi_{i})/\left \vert
k_{i}\bar{h}_{i}\right \vert ,$ which is the desired finite time.
\end{remark}

\begin{remark}
We will compare the exponential decay rate of the error system in our analysis
and in the one presented in \cite{mk21auto}. For simplicity, we consider a
single delay $D.$ When the Hessian $H$ is uncertain, Theorem \ref{theorem1}
allows a decay rate $\delta=\min_{i=1,\ldots,n}\left \vert k_{i}\bar{h}%
_{i}\right \vert \leq \frac{1}{\mathrm{e}D}$ for all $\varepsilon \in
(0,\varepsilon^{\ast}].$ In \cite{mk21auto}, the decay rate can be
approximated by $\delta=\min \left \{  \sqrt{\bar{\omega}/2},\sqrt{5/\left(
18\left(  D+1\right)  \right)  }\right \}  $ with $\bar{\omega}$ satisfying
$3\sqrt{n}\bar{\omega}D<1$ for unknown Hessian $H$, but only for
$\varepsilon \rightarrow0.$ Note that, for high dimension $n$ and not large
delay $D,$ our bound becomes larger for appropriate $k_{i}.$ For a large delay
$D$, our bound may be smaller than the one of \cite{mk21auto}.
\end{remark}

\begin{remark}
\label{remark3}We give a detailed discussion about the effect of tuning
parameters on the decay rate $\delta_{i},$ upper bound $\varepsilon^{\ast}$
and UB. For simplicity, we choose $\bar{\sigma}_{i}=\bar{\sigma},$ $k_{i}=k$
$(i=1,\ldots,n).$ For given $\kappa,$ $\mu,$ $a_{i}$ and $\bar{\sigma}%
>\bar{\sigma}_{0i}>0,$ it is clear that $\Phi_{1}^{i}$ and $\Phi_{2}^{i}$ in
(\ref{eq24a}) are increasing functions w.r.t $\left \vert k\right \vert $ and
$\varepsilon^{\ast}.$ Therefore, $\varepsilon^{\ast}$ decreases as $\left \vert
k\right \vert $ increases. On the other hand, the decay rate $\delta
_{i}=\left \vert k\bar{h}_{i}\right \vert $ increases as $\left \vert
k\right \vert $ increases. So we can adjust the gain $k$ to balance the decay
rate $\delta_{i}$ and $\varepsilon^{\ast}.$ Finally, for given available
tuning parameters $\kappa,$ $\mu,$ $\varepsilon,$ $\bar{\sigma}_{0i},$ $k_{i}$
and $a_{i},$ we explain how to find the UB as small as possible. Note that
$\Omega_{i}$ in (\ref{eq74}) are increasing functions of $\bar{\sigma}.$
Similar to the arguments in Remark 2 in \cite{yf22tac}, for the above given
parameters, we first solve the equations $\Phi_{2}^{i}=0$ ($\varepsilon^{\ast
}=\varepsilon$) to find the smallest positive $\bar{\sigma},$ and then
substitute it into (\ref{eq74}) to get values of $\Omega_{i}$ (denoted as
$\mathrm{\bar{B}}_{i1}$). If $\mathrm{\bar{B}}_{i1}<\bar{\sigma}_{0i}-\beta$
with some $\beta \in(0,\bar{\sigma}_{0i})$, from Remark \ref{remark2} it
follows that there exists a finite time $T_{i1}>0$ such that $|\tilde
{\vartheta}_{i}(t)|<\mathrm{\bar{B}}_{i1}+\beta,$ $\forall t\geq T_{i1}.$
Resetting $\bar{\sigma}_{0i}=\mathrm{\bar{B}}_{i1}+\beta$, repeat the
calculation procedure to obtain new smaller values of $\Omega_{i}$ (denoted as
$\mathrm{\bar{B}}_{i2}$)$.$ Repeating the above process until $\mathrm{\bar
{B}}_{ij}-\mathrm{\bar{B}}_{i,j+1}<\gamma$ (here $\gamma$ is a preset small
constant, for instance, $\gamma=10^{-3}$), and $\mathrm{\bar{B}}_{i,j+1}$ are
the final UB (denoted as $\mathrm{\bar{B}}_{i}$) for each $\tilde{\vartheta
}_{i}$ as we wanted. Finally, via part (ii) of Theorem \ref{theorem1}, a small
UB for $\tilde{\theta}_{i}$ (denoted as $\mathrm{B}_{i}$) can be calculated as
$\mathrm{B}_{i}=\sum_{j=1}^{n}\left \vert b_{ji}\right \vert \mathrm{\bar{B}%
}_{j}.$
\end{remark}

Finally, Theorem \ref{theorem1} guarantees for any delays $D$ and
$D_{i}^{\mathrm{in}}$ semi-global convergence for small enough $\varepsilon
^{\ast},$ $\mu$ and $\kappa$:

\begin{corollary}
\label{corollary0} Assume that \textbf{A1-A5} hold. Given any $D>0$ and
$D_{i}^{\mathrm{in}}>0$ $(i=1,\ldots,n)$ and $\bar{\sigma}_{0i}$ $>0$ (also
$\sigma_{0i}>0$) $(i=1,\ldots,n),$ and choosing $k_{i}$ $(i=1,\ldots,n)$ to
satisfy (\ref{eq23a}), the ES algorithm converges for small enough
$\varepsilon^{\ast},$ $\mu$ and $\kappa.$
\end{corollary}

\subsection{Single-variable static map}

For the single-variable static map%
\begin{equation}
\left.  Q(\theta(t-D_{\mathrm{out}}(t)))=Q^{\ast}+\frac{h}{2}[\theta(t-D_{\mathrm{out}}(t))-\theta^{\ast}%
]^{2},\right.  \label{eq75b}%
\end{equation}
where $D_{\mathrm{out}}(t)=D+\Delta D(t)$
($D>0$ is large and known, $\left \vert \Delta D(t)\right \vert \leq \rho$),
we let \textbf{A3 }be satisfied, $h$ be unknown scalar satisfying%
\begin{equation}
h_{\mathrm{m}}\leq \left \vert h\right \vert \leq h_{\mathrm{M}} \label{eq75}%
\end{equation}
with known positive bounds $h_{\mathrm{m}}$ and $h_{\mathrm{M}},$ and the
extremum point $\theta^{\ast}$ satisfy the following assumption:

\textbf{A2'.} The extremum point $\theta^{\ast}\in \mathbf{R}$ to be sought is
uncertain from a known interval $\theta^{\ast}\in \lbrack \underline{\theta
}^{\ast},\bar{\theta}^{\ast}]$ with $\left \vert \bar{\theta}^{\ast}%
-\underline{\theta}^{\ast}\right \vert =\sigma_{0}>0.$

We also consider the ES with large known input delay $D^{\mathrm{in}}>0.$ Let the dither signals $S(t)$ and
$M(t)$ satisfy
\[
\left.  S(t)=a\sin(\omega \left(  t+D^{\mathrm{in}}+D\right)  ),M(t)=\frac
{2}{a}\sin(\omega t)\right.
\]
with $\omega=\frac{2\pi}{\varepsilon},$ here $\varepsilon$ is a small
parameter. For the stability analysis, we also choose $\rho=\mu \varepsilon$
with a small parameter $\mu \geq0$ and define%
\begin{equation}
\left.
\begin{array}
[c]{l}%
D(t)=D^{\mathrm{in}}+D+\Delta D(t),\\
D_{\mathrm{M}}=D_{\mathrm{M}}(\varepsilon)=D^{\mathrm{in}}+D+\mu \varepsilon.
\end{array}
\right.  \label{eq75a}%
\end{equation}
Then with the ES algorithm (\ref{eq67}) for $n=1$, the estimation error
$\tilde{\theta}(t)$ is governed by \vspace{-0.2cm}
\begin{equation}
\left.
\begin{array}
[c]{l}%
\dot{\tilde{\theta}}(t)=kM(t)Q^{\ast}+\frac{1}{2}khM(t)\\
\times \lbrack S(t)+\tilde{\theta}(t-D(t))]^{2}+w(t),\text{ }t\geq
D_{\text{\textrm{M}}},\\
\tilde{\theta}(t)=\tilde{\theta}(D_{\text{\textrm{M}}}),\text{ }t\in
\lbrack0,D_{\text{\textrm{M}}}]
\end{array}
\right.  \label{eq76a}%
\end{equation}
with $w(t)$ satisfying updated (\ref{eq33a}). Applying the time-delay approach
to averaging of (\ref{eq76a}), we finally arrive at \vspace{-0.2cm}
\[
\left.
\begin{array}
[c]{l}%
\dot{z}(t)=khz(t-D(t))\\
\text{ \  \  \  \  \  \  \  \ }+\bar{w}(t),\text{ }t\geq D_{\text{\textrm{M}}%
}+\varepsilon,\text{ }i=1,\ldots,n,\\
\bar{w}(t)=khG(t-D(t))-Y_{1}(t)\\
\text{ \  \  \  \  \  \  \  \ }-Y_{2}(t)+w(t)
\end{array}
\right.
\]
with $\{z(t),G(t),Y_{1}(t),Y_{2}(t),w(t)\}$ satisfying updated (\ref{eq52a}),
(\ref{eq20}), (\ref{eq20a}) and (\ref{eq33a}), respectively.

We will find $k$ from the inequality
\begin{equation}
\left.  \left \vert k\right \vert h_{\mathrm{M}}\left(  D^{\mathrm{in}%
}+D\right)  -\frac{1}{\mathrm{e}}<0,\right.  \label{eq78}%
\end{equation}
which guarantees the exponential stability with a decay rate
\begin{equation}
\left.  \delta=\left \vert k\right \vert h_{\mathrm{M}}<\frac{1}{\mathrm{e}%
\left(  D^{\mathrm{in}}+D\right)  }\right.  \label{eq78a}%
\end{equation}
of the averaged system with $\rho=0:$
\begin{equation}
\left.  \dot{z}(t)=khz(t-D^{\mathrm{in}}-D).\right.  \label{eq78b}%
\end{equation}
Following the arguments of Theorem \ref{theorem1}, we have the following corollary.

\begin{corollary}
\label{corollary2}Assume that \textbf{A2'},\textbf{ A3} and \textbf{(}%
\ref{eq75}\textbf{) }hold\textbf{. }let $k$ satisfy (\ref{eq78}) and
$D_{\mathrm{M}}$ be defined in (\ref{eq75a}). Given tuning parameters $a$ and
$\mu \geq0$ as well as $\sigma>\sigma_{0}>0$, let there exists $\varepsilon
^{\ast}>0$ that satisfy \vspace{-0.2cm}
\begin{equation}
\left.
\begin{array}
[c]{l}%
\Phi_{1}=\left \vert k\right \vert h_{\mathrm{M}}D_{\mathrm{M}}(\varepsilon
^{\ast})-\frac{1}{\mathrm{e}}\leq0\mathbf{,}\\
\Phi_{2}=\mathrm{e}^{\left \vert k\right \vert h_{\mathrm{M}}D_{\mathrm{M}%
}(\varepsilon^{\ast})}\left[  \sigma_{0}+\frac{3\varepsilon^{\ast}\Delta}%
{2}+W(\varepsilon^{\ast},\mu)\right] \\
\text{ \  \  \  \  \  \  \ }+\frac{\varepsilon^{\ast}\Delta}{2}-\sigma<0\mathbf{,}%
\end{array}
\right.  \label{eq76}%
\end{equation}
where%
\begin{equation}
\left.
\begin{array}
[c]{l}%
\Delta=\frac{2\left \vert k\right \vert }{\left \vert a\right \vert }\left[
Q_{\mathrm{M}}^{\ast}+\frac{h_{\mathrm{M}}}{2}\left(  \sigma+\left \vert
a\right \vert \right)  ^{2}\right]  ,\\
W(\varepsilon,\mu)=\varepsilon \frac{\left[  \left \vert a\right \vert +\left(
2+8\mu \right)  \left(  \sigma+\left \vert a\right \vert \right)  \right]
\Delta}{2\left \vert a\right \vert }\\
\text{ \  \  \  \  \  \  \  \  \  \  \ }+\mu4\pi \left(  \left \vert a\right \vert +\mu
\pi \left \vert a\right \vert +\sigma \right)  .
\end{array}
\right.  \label{eq109}%
\end{equation}
Then for all $\varepsilon \in(0,\varepsilon^{\ast}]$, the solution
$\tilde{\theta}(t)$ of system (\ref{eq76a}) with $|\tilde{\theta
}(D_{\mathrm{M}})|\leq \sigma_{0}$ satisfies%
\begin{equation}
\left.
\begin{array}
[c]{l}%
|\tilde{\theta}(t)|<\left \vert \tilde{\theta}(D_{\mathrm{M}})\right \vert
+\varepsilon \Delta<\sigma,\text{ }t\in \lbrack D_{\mathrm{M}},D_{\mathrm{M}%
}+\varepsilon],\\
|\tilde{\theta}(t)|<\left(  1+\left \vert k\right \vert h_{\mathrm{M}%
}D_{\mathrm{M}}\right)  \left[  \left \vert \tilde{\theta}(D_{\mathrm{M}%
})\right \vert +\frac{3\varepsilon \Delta}{2}\right] \\
\text{ \  \  \  \  \  \  \  \  \ }+\left \vert k\right \vert h_{\mathrm{M}}%
D_{\mathrm{M}}W(\varepsilon,\mu)+\frac{\varepsilon \Delta}{2}\\
\text{ \  \  \  \  \  \  \ }<\sigma,\text{ }t\in \lbrack D_{\mathrm{M}}%
+\varepsilon,2D_{\mathrm{M}}+\varepsilon].\\
|\tilde{\theta}(t)|<\mathrm{e}^{-\left \vert k\right \vert h_{\mathrm{m}}\left(
t-2D_{\mathrm{M}}-\varepsilon \right)  }\mathrm{e}^{\left \vert k\right \vert
h_{\mathrm{M}}D_{\mathrm{M}}}\left[  \left \vert \tilde{\vartheta
}(D_{\mathrm{M}})\right \vert +\frac{3\varepsilon \Delta}{2}\right] \\
\text{ \  \  \  \  \  \  \  \  \ }+\mathrm{e}^{\left \vert k\right \vert h_{\mathrm{M}%
}D_{\mathrm{M}}}W(\varepsilon,\mu)+\frac{\varepsilon \Delta}{2}\\
\text{ \  \  \  \  \  \  \ }<\sigma,\text{ }t\geq2D_{\mathrm{M}}+\varepsilon.
\end{array}
\right.  \label{eq81}%
\end{equation}
This solution is exponentially attracted to the interval%
\[
\left.  \left \vert \tilde{\theta}(t)\right \vert <\mathrm{e}^{\left \vert
k\right \vert h_{\mathrm{M}}D_{\mathrm{M}}}W(\varepsilon,\mu)+\frac
{\varepsilon \Delta}{2}\right.
\]
with a decay rate $\delta=\left \vert k\right \vert h_{\mathrm{m}}.$
\end{corollary}

\begin{remark}
For the non-delayed 1D map with $D(t)=0,$ (\ref{eq76}) and (\ref{eq81}) are
reduced to%
\[
\left.  \Phi=\sigma_{0}+\frac{\varepsilon^{\ast}\Delta \left(  7\left \vert
a\right \vert +2\sigma \right)  }{2\left \vert a\right \vert }-\sigma<0,\right.
\]
and%
\[
\left.
\begin{array}
[c]{l}%
\vert \tilde{\theta}(t) \vert<\left \vert \tilde{\theta}(0)\right \vert
+\varepsilon \Delta<\sigma,\text{ }t\in \lbrack0,\varepsilon],\\
\vert \tilde{\theta}(t) \vert<\mathrm{e}^{-\left \vert K\right \vert H_{m}\left(
t-\varepsilon \right)  }\left(  \left \vert \tilde{\theta}(0)\right \vert
+\frac{3\varepsilon \Delta}{2}\right) \\
\text{ \  \  \  \  \  \  \  \  \  \ }+\frac{\varepsilon \Delta \left(  2\left \vert
a\right \vert +\sigma \right)  }{\left \vert a\right \vert }<\sigma,\text{ }%
t\geq \varepsilon,
\end{array}
\right.
\]
which coincide with the results in Remark 5 of \cite{yf22tac}. Hence,
Corollary \ref{corollary2} recovers the results of \cite{yf22tac} for
non-delayed case.
\end{remark}

\subsection{Examples\label{sec2.2}}

Example 3.1. ($1\mathrm{D}$\textit{ static map}) Consider the single-input map
(\ref{eq1}) for $n=1$ with%
\begin{equation}
\left.
\begin{array}
[c]{l}%
\left \vert Q^{\ast}\right \vert \leq Q_{\mathrm{M}}^{\ast}=0.5,\\
H=\bar{H}+\Delta H,\bar{H}=2,\left \vert \Delta H\right \vert \leq \kappa.
\end{array}
\right.  \label{eq22}%
\end{equation}
Let the delays in (\ref{eq75a}) satisfy $D^{\mathrm{in}}=D=1.$ We select the
tuning parameters of the gradient-based ES as
\[
a=0.3, { \ } k=-0.003.
\]
To
calculate UB, we choose $\beta=\gamma=0.001$ following Remark \ref{remark3}.
The results that follow from Theorem \ref{theorem1} and Corollary
\ref{corollary2} are shown in Table \ref{table1}. By comparing the data, we
find that for the same values of $\sigma_{0},$ $\sigma$ and $\mu,$ the results
in Corollary \ref{corollary2} allow larger upper bounds $\kappa$ and
$\varepsilon^{\ast}$ (smaller lower bound frequency $\omega^{\ast}$) as well
as smaller UB than those in Theorem \ref{theorem1}.

For the numerical simulations, we choose $\varepsilon=0.74,$ $\mu=0.01,$
$\Delta D(t)=\rho \sin t$ with $\rho=\varepsilon \mu=0.74\cdot10^{-2},$
$Q^{\ast}=0,$ $H=2$ and the other parameters as shown above. Under the initial
condition $\hat{\theta}(s)=0.5,$ $s\in \lbrack0,2+\rho],$ the simulation result
is shown in Fig. \ref{fig0}, from which we can see that the value of UB shown
in Table \ref{table1} is confirmed.

\begin{table}[h]
\caption{Example 3.1: comparison of $\delta$, $\kappa$, $\varepsilon^{\ast}$, $\omega
^{\ast}=2\pi/\varepsilon^{\ast}$ and UB}%
\label{table1}%
\center  {\tiny { \setlength{\tabcolsep}{3pt}
\begin{tabular}
[c]{|l|c|c|c|c|c|c|c|c|}\hline
ES: sine & $\sigma_{0}$ & $\sigma$  & $\mu$ &$\delta$ & $\kappa$ & $\varepsilon^{\ast}$ & $\omega
^{\ast}$ & UB\\ \hline
Corollary \ref{corollary2} & 0.5 & 1  & 0.01 & $0.003$ & 1 & 0.74 & 8.49 & 0.115\\ \hline
Theorem \ref{theorem1} & 0.5 & 1  & 0.01 & $0.006$ & 0.47 & 0.1 & 62.83 & 0.358\\ \hline
\end{tabular}
} }\end{table}

\begin{figure}[ptb]
\centering
\includegraphics[width=0.4\textwidth]{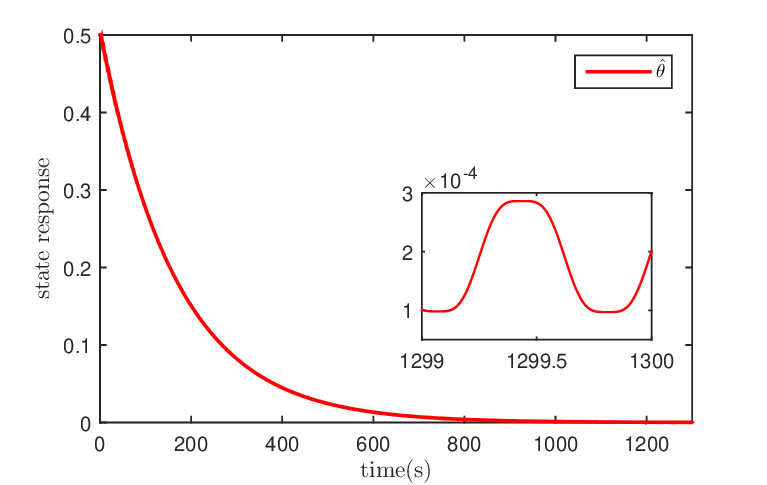}\caption{Example 3.1: Trajectory of the estimate $\hat{\theta}$}%
\label{fig0}%
\end{figure}

Example 3.2. ($2\mathrm{D}$\textit{ static map}) Consider an autonomous vehicle in an
environment without GPS orientation \cite{sk17book}. The goal is to reach the
the origin, i.e. location of the stationary minimum of a measured distance
function%
\[
\left.
\begin{array}
[c]{l}%
J(x(t),y(t))=Q^{\ast}+\frac{1}{2}\left[
\begin{array}
[c]{cc}%
x(t) & y(t)
\end{array}
\right]  H\left[
\begin{array}
[c]{c}%
x(t)\\
y(t)
\end{array}
\right] \\
\text{ \  \  \  \  \  \  \  \  \  \  \  \  \  \  \  \ }=x^{2}(t)+y^{2}(t),
\end{array}
\right.
\]
where $Q^{\ast}=0$ and $H=[2,0;0,2].$ We employ the classical ES with
measurement and input delays:
\[
\left.
\begin{array}
[c]{l}%
x(t)=\hat{x}(t)+a_{1}\sin \left(  \frac{2\pi l_{1}}{\varepsilon}t\right)  ,\\
y(t)=\hat{y}(t)+a_{2}\sin \left(  \frac{2\pi l_{2}}{\varepsilon}t\right)  ,\\
\dot{\hat{x}}(t)=\frac{2k_{1}}{a_{1}}\sin \left(  \frac{2\pi l_{1}}%
{\varepsilon}t\right)  J(t-D_{1}(t)),\\
\dot{\hat{y}}(t)=\frac{2k_{2}}{a_{2}}\sin \left(  \frac{2\pi l_{2}}%
{\varepsilon}t\right)  J(t-D_{2}(t)),
\end{array}
\right.
\]
where $D_{i}(t)=D_{i}^{\mathrm{in}}+D+\Delta D(t)$ ($i=1,2$) with $\left \vert
\Delta D(t)\right \vert \leq \rho=\mu \varepsilon$. We select the delays and
tuning parameters as%
\begin{equation}
\left.
\begin{array}
[c]{l}%
D_{1}^{\mathrm{in}}=0.5,\text{ }D_{2}^{\mathrm{in}}=1,\text{ }D=1,\\
k_{1}=k_{2}=-0.003,\text{ }a_{1}=a_{2}=0.3.
\end{array}
\right.  \label{eq59}%
\end{equation}
Following (\ref{eq34}), we choose%
\begin{equation}
\left.  \varepsilon=\frac{1}{2q},\text{ }q\in \mathbf{N}.\right.  \label{eq60a}%
\end{equation}
To calculate UB, we also choose $\beta=\gamma=0.001$ following Remark
\ref{remark3}. Verifying conditions of Theorem \ref{theorem1} ($[\tilde
{x}(t),\tilde{y}(t)]^{\mathrm{T}}=[\tilde{\vartheta}_{1}(t),\tilde{\vartheta
}_{2}(t)]^{\mathrm{T}}$ since $U=I_{2}$), we arrive at the results shown in Table
\ref{table2}, in which the UB ($\mathrm{\bar{B}}_{i},$ $i=1,2$) corresponds to
$\varepsilon=1/4$ $(<\varepsilon^{\ast}=0.3).$

For the numerical simulations, we choose $\varepsilon=1/4,$ $\mu=0.005,$
$\Delta D(t)=\rho \sin t$ with $\rho=\varepsilon \mu=0.125\cdot10^{-2}$ and the
other parameters as shown in (\ref{eq59}). Under the initial condition
$\hat{x}(s)=-\hat{y}(s)=0.5,$ $s\in \lbrack0,2+\rho]$ (thus $[\tilde
{x}(s),\tilde{y}(s)]^{\mathrm{T}}=[0.5,-0.5]^{\mathrm{T}},$ $s\in
\lbrack0,2+\rho]$)$,$ the simulation results are shown in Fig. \ref{fig1},
from which we can see that the values of UB shown in Table \ref{table2} are confirmed.

\begin{table}[h]
\caption{Example 3.2 with (\ref{eq59}): The values of $\mu$, $\varepsilon^{\ast}$, $\omega^{\ast}=
2\pi/\varepsilon^{\ast}$, $\delta_{i}$ (i=1,2) and UB ($\mathrm{\bar{B}}_{i},$
$i=1,2$)}%
\label{table2}%
\center  {\tiny { \setlength{\tabcolsep}{3pt}
\begin{tabular}
[c]{|l|c|c|c|c|c|c|c|}\hline
ES: sine & $\bar{\sigma}_{01}/\bar{\sigma}_{02}$ & $\bar{\sigma}_{1}/\bar{\sigma}_{2}$ & $\mu$
& $\varepsilon^{\ast}$ & $\omega^{\ast}$ & $\delta_{1}=\delta_{2}$ & $\mathrm{\bar{B}}_{1}/\mathrm{\bar{B}}_{2}$\\ \hline
Theorem \ref{theorem1} & 0.5/0.5 & 1/1  &$0.005$  &0.30 & 20.94 & $0.006$ & $0.124/0.124$ \\ \hline
\end{tabular}
} }\end{table}

\begin{figure}[ptb]
\centering
\includegraphics[width=0.4\textwidth]{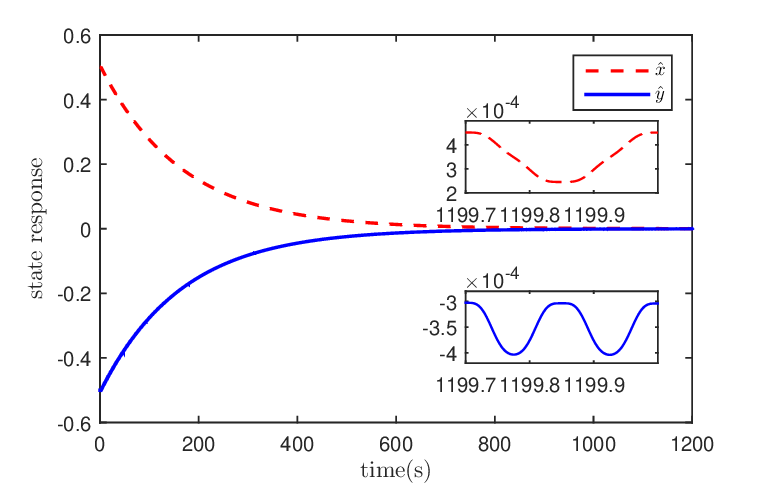}\caption{Example 3.2: Trajectories of the estimate ($\hat{x},\hat{y}$)}%
\label{fig1}%
\end{figure}

Example 3.3. ($2\mathrm{D}$\textit{ static map}) Following \cite{mk21auto,okt17tac}, we
consider the static map (\ref{eq1}) with%
\begin{equation}
\left.  Q^{\ast}=1,\text{ }\theta^{\ast}=\left[
\begin{array}
[c]{c}%
0\\
1
\end{array}
\right]  ,\text{ }\bar{H}=\left[
\begin{array}
[c]{cc}%
-2 & -2\\
-2 & -4
\end{array}
\right]  ,\left \Vert \Delta H\right \Vert \leq \kappa.\right.  \label{eq82b}%
\end{equation}
We find that%
\[
U=\left[
\begin{array}
[c]{cc}%
0.5257 & 0.8507\\
-0.8507 & 0.5257
\end{array}
\right]
\]
satisfies $U^{\mathrm{T}}U=I_{2}$, $\left \Vert U\right \Vert =1$ and%
\[
UHU^{\mathrm{T}}=\left[
\begin{array}
[c]{cc}%
-5.2361 & 0\\
0 & -0.7639
\end{array}
\right]  .
\]
We consider two cases: $\kappa=0$ and $\kappa=0.2.$ We choose the delays
$D_{i}(t)=D_{i}^{\mathrm{in}}+D+\varepsilon \Delta D(t)$ ($i=1,2$) and the
tuning parameters as%
\begin{equation}
\left.
\begin{array}
[c]{l}%
D_{1}^{\mathrm{in}}=0.5,\text{ }D_{2}^{\mathrm{in}}=1.5,\text{ }D=1,\\
\left \vert \Delta D(t)\right \vert \leq \rho=\mu \varepsilon,\text{ }a_{1}%
=a_{2}=0.3,\\
k_{1}=0.2\cdot10^{-3},\text{ }k_{2}=0.135\cdot10^{-2}.
\end{array}
\right.  \label{eq60}%
\end{equation}
Following (\ref{eq34}), we choose $\varepsilon$ as in (\ref{eq60a}). Via
Remark \ref{remark3}, we choose $\beta=\gamma=0.001$ to calculate UB
($\mathrm{\bar{B}}_{i}$ for $\tilde{\vartheta}_{i}(t)$, $\mathrm{B}_{i}$ for
$\tilde{\theta}_{i}(t),$ $i=1,2$). The results that follow from Theorem
\ref{theorem1} for known ($\kappa=0$) and uncertain ($\kappa=0.2$) $H$ are
shown in Table \ref{table3}, in which UB ($\mathrm{\bar{B}}_{i},$ $i=1,2$)
corresponds to $\varepsilon=1/4$ $(<\varepsilon^{\ast}=0.49)$ and
$\varepsilon=1/8$ $(<\varepsilon^{\ast}=0.016)$ for known and uncertain $H,$ respectively.

For the numerical simulations, we choose $\varepsilon=1/4,$ $\mu=0.003,$
$\Delta D(t)=\rho \sin t$ with $\rho=\varepsilon \mu=0.75\cdot10^{-2}$ and the
other parameter values as shown in (\ref{eq82b}) and (\ref{eq60}). For the
real-time estimate $\hat{\theta}(t)$ with the estimation error $\tilde{\theta
}(t),$ following part (ii) of Theorem \ref{theorem1}, we solve
\begin{equation}
\left.
\begin{array}
[c]{l}%
0.5257\sigma_{01}+0.8507\sigma_{02}=0.5,\\
0.8507\sigma_{01}+0.5257\sigma_{02}=0.5,
\end{array}
\right.  \label{eq16a}%
\end{equation}
to obtain $\sigma_{01}=\sigma_{02}=0.3633,$ and the UB ($\mathrm{B}_{i},$
$i=1,2$) can be calculated as%
\begin{equation}
\left.
\begin{array}
[c]{l}%
\mathrm{B}_{1}=0.5257\cdot0.023+0.8507\cdot0.16=0.1482,\\
\mathrm{B}_{2}=0.8507\cdot0.023+0.5257\cdot0.16=0.1037.
\end{array}
\right.  \label{eq16}%
\end{equation}
Under the initial condition $\hat{\theta}(s)=[0.3,0.7]^{\mathrm{T}},$
$s\in \lbrack0,2+\rho]$ (thus $\tilde{\theta}(s)=\hat{\theta}(s)-\theta^{\ast
}=[0.3633,-0.3633]^{\mathrm{T}},$ $s\in \lbrack0,2+\rho]$), the simulation
results are shown in Fig. \ref{fig2}, from which we can see that the values of
UB in (\ref{eq16}) are confirmed.

\begin{table}[h]
\caption{Example 3.3 with (\ref{eq60}): The values of $\mu$, $\varepsilon^{\ast}$, $\omega^{\ast}%
=2\pi/\varepsilon^{\ast}$, $\delta_{i}$ (i=1,2) and UB ($\mathrm{\bar{B}}%
_{i},$ $i=1,2$))}%
\label{table3}
\center  {\tiny { \setlength{\tabcolsep}{3pt}
\begin{tabular}
[c]{|l|c|c|c|c|c|c|c|}\hline
ES: sine & $\bar{\sigma}_{01}/\bar{\sigma}_{02}$ & $\bar{\sigma}_{1}/\bar{\sigma}_{2}$  & $\mu$
& $\varepsilon^{\ast}$ &$\omega^{\ast}$ & $\delta_{1}=\delta_{2}$ & $\mathrm{\bar{B}}_{1}/\mathrm{\bar{B}}_{2}$\\ \hline
Known $H$ & 0.5/0.5 & 0.6/1.0 &$0.003$  &0.49 & 12.82 & $0.001$ & 0.023/0.16\\ \hline
Uncertain $H$ & 0.5/0.5 & 0.6/1.0  &$0.001$  &0.016 &392.70& $0.001$ &0.019/0.13\\ \hline
\end{tabular}
} }\end{table}

\begin{figure}[ptb]
\centering
\includegraphics[width=0.4\textwidth]{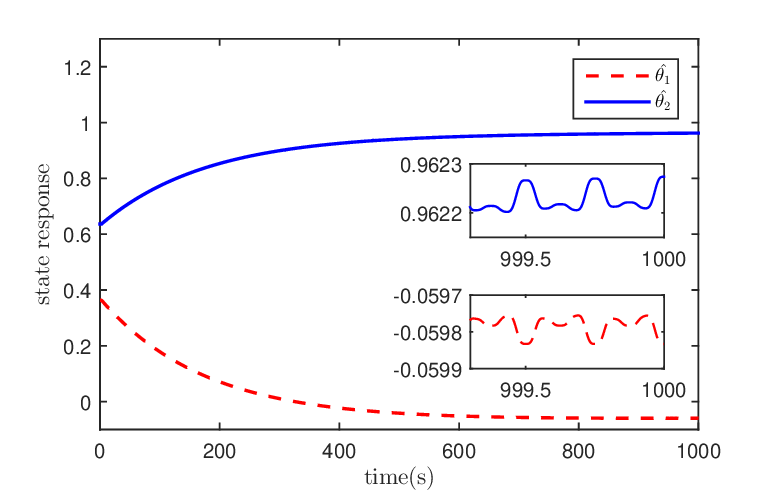}\caption{Example 3.3: Trajectories of the estimate ($\hat{\theta_{1}},\hat{\theta_{2}}$)}
\label{fig2}%
\end{figure}

\section{Sampled-data ES with Large Distinct Delays\label{sec4}}

\subsection{Multi-variable static map}

In this section, we further study the sampled-data ES with square wave
dithers, known large and constant input delays $D_{i}>0$ $(i=1,\ldots,n)$ and
periodic fast sampling. Here we assume that $D_{i}/D_{j}$ $(i,j=1,\ldots,n)$
are rational, then there exist some $m_{i}\in \mathbf{N}$ $(i=1,\ldots,n)$ such
that
\begin{equation}
m_{i}D_{j}=m_{j}D_{i},\text{ }\forall i,j=1,\ldots,n. \label{eq101}
\end{equation}
We also choose the corresponding $m_{i}$ $(i=1,\ldots,n)$ as small as possible.

Consider the multi-variable static map $Q(\theta(t))$ in (\ref{eq1}) without
measurement delay (namely, $D_{\mathrm{out}}(t)=0$), which is measured at the
sampling instants $s_{p}$ satisfying $s_{p}=\frac{p\varepsilon}{2^{n}},$
$p\in \mathbf{Z}_{+}$ with a small parameter $\varepsilon>0.$ Let \textbf{A2-A4
}be satisfied. Then equations (\ref{eq32a})-(\ref{eq11}) hold with
$D_{\mathrm{out}}(t)=0$. Define the square wave signal $\mathrm{sq}(t)$ in the
following form \cite{zfo22tac}
\begin{equation}
\left.  \mathrm{sq}(t)=\left \{
\begin{array}
[c]{cl}%
1, & t\in \left[  p,p+\frac{1}{2}\right)  ,\\
-1, & t\in \left[  p+\frac{1}{2},1\right)  ,
\end{array}
\right.  p\in \mathbf{Z}_{+}.\right.  \label{eq55d}
\end{equation}
We choose the adaptation gain $K$ as (\ref{eq40}) such that $K\mathcal{\bar
{H}}$ is Hurwitz, and
\begin{equation}
\left.
\begin{array}
[c]{l}%
S(t)=[S_{1}(t),\ldots,S_{n}(t)]^{\mathrm{T}},\text{ }S_{i}(t)=a_{i}%
\mathrm{sq}\left(  \frac{l_{i}t}{\varepsilon}\right)  ,\\
M(t)=[M_{1}(t),\ldots,M_{n}(t)]^{\mathrm{T}},\text{ }M_{i}(t)=\frac{1}{a_{i}%
}\mathrm{sq}\left(  \frac{l_{i}t}{\varepsilon}\right)
\end{array}
\right.  \label{eq93}
\end{equation}
with $l_{i}=2^{i-1},$ $l_{i}\neq l_{j},i\neq j$ and $a_{i}$ are non-zero real numbers.
By choosing
\begin{equation}
\left.  \varepsilon=\frac{D_{j}}{qm_{j}},\text{ }q\in \mathbf{N}\right.
\label{eq85}
\end{equation}
and using arguments of Lemma \ref{lemma2}, we find $\mathrm{sq}(l_{i}%
(t-D_{j})/\varepsilon)=\mathrm{sq}(l_{i}t/\varepsilon),$ $(i,j=1,\ldots,n).$
Then for $\forall i,j=1,\ldots,n,$
\begin{equation}
\left.  S_{i}(t-D_{j})=S_{i}(t),\text{ }M_{i}(t-D_{j})=M_{i}(t).\right.
\label{eq87}%
\end{equation}
Let $D_{\text{\textrm{M}}}=\max_{i=1,\ldots,n}D_{i}$ and $m=\max
_{i=1,\ldots,n}m_{i}.$ Via (\ref{eq101}) and (\ref{eq85}) we have
$D_{\text{\textrm{M}}}-D_{i}=\frac{m-m_{i}}{m_{i}}D_{i}=q_{i}\cdot
\frac{\varepsilon}{2^{n}}$ and $q_{i}=2^{n}(m-m_{i})q,$ $q\in \mathbf{N}.$

When each signal $\vartheta_{i}$ at the control update instant $t_{p}^{i}$
experiences a constant delay $D_{i}>0$ ($t_{p}^{i}=s_{p}+D_{i}$), the
gradient-based sampled-data ES algorithm can be designed as:
\begin{equation}
\left.  \dot{\hat{\vartheta}}_{i}(t)=\left \{
\begin{array}
[c]{ll}%
0, & t\in \lbrack0,D_{\text{\textrm{M}}}),\\
k_{i}M_{i}(s_{p})y(s_{p}), & t\in \lbrack t_{p}^{i},t_{p+1}^{i}),
\end{array}
\right.  \right.  \label{eq55c}
\end{equation}
where $p\geq q_{i},p\in \mathbf{Z}_{+},$ $y(t)=Q(U^{\mathrm{T}}\vartheta(t))$
with $Q(U^{\mathrm{T}}\vartheta(t))$ given by (\ref{eq26}) ($D_{\mathrm{out}%
}(t)=0$) and $\vartheta(t)=\hat{\vartheta}(t)+S(t).$ Via (\ref{eq26}),
(\ref{eq55b}) and (\ref{eq55c}), the estimation error is governed by
\begin{equation}
\left.
\begin{array}
[c]{l}%
\dot{\tilde{\vartheta}}_{i}(t)=k_{i}M_{i}(s_{p})Q^{\ast}+\frac{1}{2}k_{i}%
M_{i}(s_{p})[\tilde{\vartheta}(s_{p})+S(s_{p})]^{\mathrm{T}}\\
\times \mathcal{H}[\tilde{\vartheta}(s_{p})+S(s_{p})],\text{ }t\in \lbrack
t_{p}^{i},t_{p+1}^{i}),\text{ }p\geq q_{i},\text{ }p\in \mathbf{Z}_{+},\\
\tilde{\vartheta}_{i}(t)=\tilde{\vartheta}_{i}(D_{\text{\textrm{M}}}),\text{
}t\in \lbrack0,D_{\text{\textrm{M}}}].
\end{array}
\right.  \label{eq86}
\end{equation}
Following the time-delay approach to sampled-data control
\cite{fri14book,zfo22tac}, denote
\begin{equation}
\left.
\begin{array}
[c]{l}%
\eta_{i}(t)=t-D_{i}-s_{p},\text{ }t\in \lbrack t_{p}^{i},t_{p+1}^{i}),\\
0\leq \eta_{i}(t)<\frac{\varepsilon}{2^{n}},\text{ }i=1,\ldots,n\mathbf{.}%
\end{array}
\right.  \label{eq83}
\end{equation}
Note that for $\forall i,j=1,\ldots,n,$
\begin{equation}
\left.  \mathrm{sq}\left(  \frac{l_{j}s_{p}}{\varepsilon}\right)
=\mathrm{sq}\left(  \frac{l_{j}\left(  t-D_{i}\right)  }{\varepsilon}\right)
,\text{ }t\in \lbrack t_{p}^{i},t_{p+1}^{i}),\right.  \label{eq43}
\end{equation}
which with (\ref{eq87}) implies that
\begin{equation}
\left.
\begin{array}
[c]{l}%
M_{j}(s_{p})=M_{j}(t-D_{i})=M_{j}(t),\text{ }t\in \lbrack t_{p}^{i},t_{p+1}%
^{i}),\\
S_{j}(s_{p})=S_{j}(t-D_{i})=S_{j}(t),\text{ }t\in \lbrack t_{p}^{i},t_{p+1}%
^{i}).
\end{array}
\right.  \label{eq84}
\end{equation}
Taking into account (\ref{eq83}) and (\ref{eq84}), the dynamics (\ref{eq86})
becomes
\begin{equation}
\left.
\begin{array}
[c]{l}%
\dot{\tilde{\vartheta}}_{i}(t)=k_{i}M_{i}(t)Q^{\ast}+\frac{1}{2}k_{i}%
M_{i}(t)[\tilde{\vartheta}(t-D_{i}(t))+S(t)]^{\mathrm{T}}\\
\text{ \  \  \  \  \  \ }\times \mathcal{H}[\tilde{\vartheta}(t-D_{i}%
(t))+S(t)],\text{ }t\geq D_{\text{\textrm{M}}},\text{ }i=1,\ldots
,n\mathbf{,}\\
\tilde{\vartheta}_{i}(t)=\tilde{\vartheta}_{i}(D_{\text{\textrm{M}}}),\text{
}t\in \lbrack0,D_{\text{\textrm{M}}}]
\end{array}
\right.  \label{eq84a}%
\end{equation}
with
\begin{equation}
D_{i}(t)=D_{i}+\eta_{i}(t). \label{eq84b}
\end{equation}

Recently, for the case of the single input delay $D_{i}=D>0$ ($i=1,\ldots,n$)
and the dimension $n=1,2,$ motivated by \cite{fz20auto}, a time-delay approach
for the stability analysis of sampled-data ES algorithms was introduced in
\cite{zfo21cdc,zfo22tac}. In the latter papers, the ES dynamics was first
converted into a neutral type model with a nominal delay-free system, and then
the L-K method was used to find sufficient practical stability conditions in
the form of LMIs. Different from the treatments for (\ref{eq84a}) in
\cite{zfo21cdc,zfo22tac}, in this paper we will first transform (\ref{eq84a})
into a neutral type model with a nominal time-delay system, and further
present the resulting neutral system as a retarded one. Finally, motivated by
our previous work in \cite{yf22tac}, we will employ the variation of constants
formula to find sufficient practical stability conditions.

As in Section \ref{sec3}, we integrate from $t-\varepsilon$ to $t$ and divide
by $\varepsilon$ on both sides of (\ref{eq84a}) to obtain%
\begin{equation}
\left.
\begin{array}
[c]{l}%
\frac{1}{\varepsilon}%
{\textstyle \int \nolimits_{t-\varepsilon}^{t}}
\dot{\tilde{\vartheta}}_{i}(\tau)\mathrm{d}\tau=\frac{1}{\varepsilon}%
{\textstyle \int \nolimits_{t-\varepsilon}^{t}}
k_{i}M_{i}(\tau)Q^{\ast}\mathrm{d}\tau \\
+\frac{1}{2\varepsilon}%
{\textstyle \int \nolimits_{t-\varepsilon}^{t}}
k_{i}M_{i}(\tau)S^{\mathrm{T}}(\tau)\mathcal{H}S(\tau)\mathrm{d}\tau \\
+\frac{1}{2\varepsilon}%
{\textstyle \int \nolimits_{t-\varepsilon}^{t}}
k_{i}M_{i}(\tau)\tilde{\vartheta}^{\mathrm{T}}(\tau-D_{i}(\tau))\mathcal{H}%
\tilde{\vartheta}(\tau-D_{i}(\tau))\mathrm{d}\tau \\
+\frac{1}{\varepsilon}%
{\textstyle \int \nolimits_{t-\varepsilon}^{t}}
k_{i}M_{i}(\tau)S^{\mathrm{T}}(\tau)\mathcal{H}\tilde{\vartheta}(\tau
-D_{i}(\tau))\mathrm{d}\tau,\\
\text{ }t\geq D_{\text{\textrm{M}}}+\varepsilon,\text{ }i=1,\ldots,n.
\end{array}
\right.  \label{eq100}%
\end{equation}
In view of the definition of $\mathrm{sq}(t)$ in (\ref{eq55d}) and
$l_{i}=2^{i-1},$ $i=1,\ldots,n\mathbf{,}$ we can prove that for $\forall
i,j,k=1,\ldots,n\mathbf{,}$
\begin{equation}
\left.
\begin{array}
[c]{l}%
{\textstyle \int \nolimits_{t-\varepsilon}^{t}}
\mathrm{sq}\left(  \frac{l_{i}\tau}{\varepsilon}\right)  \mathrm{d}\tau=0,\\
\int_{t-\varepsilon}^{t}\mathrm{sq}\left(  \frac{l_{i}\tau}{\varepsilon
}\right)  \mathrm{sq}\left(  \frac{l_{j}\tau}{\varepsilon}\right)
\mathrm{sq}\left(  \frac{l_{k}\tau}{\varepsilon}\right)  \mathrm{d}%
\tau=0,\text{ }\\
\int \nolimits_{t-\varepsilon}^{t}\frac{a_{i}}{a_{j}}\mathrm{sq}\left(
\frac{l_{i}\tau}{\varepsilon}\right)  \mathrm{sq}\left(  \frac{l_{j}\tau
}{\varepsilon}\right)  \mathrm{d}\tau=\left \{
\begin{array}
[c]{cc}%
\varepsilon, & i=j,\\
0, & i\neq j.
\end{array}
\right.
\end{array}
\right.  \label{eq85a}%
\end{equation}
Due to (\ref{eq85a}), relations (\ref{eq48})-(\ref{eq51}) hold for $S(t)$ and
$M(t)$ defined by (\ref{eq93}), $\tilde{\vartheta}(t)=[\tilde{\vartheta}%
_{1}(t),\ldots,\tilde{\vartheta}_{n}(t)]^{\mathrm{T}}$ satisfying
(\ref{eq84a}) and $D_{i}(t)$ given by (\ref{eq84b}). Let $G_{i}(t)$ be of the
form of (\ref{eq20})-(\ref{eq20c}). Then we can present%
\begin{equation}
\left.
\begin{array}
[c]{l}%
\frac{1}{\varepsilon}%
{\textstyle \int \nolimits_{t-\varepsilon}^{t}}
\dot{\tilde{\vartheta}}_{i}(\tau)\mathrm{d}\tau=\frac{\mathrm{d}}{\mathrm{d}%
t}\left[  \tilde{\vartheta}_{i}(t)-G_{i}(t)\right]  .
\end{array}
\right.  \label{eq17}%
\end{equation}

Substituting (\ref{eq48})-(\ref{eq51}) and (\ref{eq17}) into (\ref{eq100}) and
employing notations (\ref{eq20a}) and (\ref{eq52a}), we finally arrive at%
\begin{equation}
\left.
\begin{array}
[c]{l}%
\dot{z}_{i}(t)=k_{i}\bar{h}_{i}z_{i}(t-D_{i}(t))\\
\text{ \  \  \  \  \  \  \  \ }+\bar{w}_{i}(t),\text{ }t\geq D_{\text{\textrm{M}}%
}+\varepsilon,\text{ }i=1,\ldots,n,\\
\bar{w}_{i}(t)=k_{i}\bar{h}_{i}G_{i}(t-D_{i}(t))-Y_{1i}(t)\\
\text{ \  \  \ }\  \  \  \  \ -Y_{2i}(t)+Y_{3i}(t).
\end{array}
\right.  \label{eq95}%
\end{equation}
Similar to the arguments in Section \ref{sec3}, we can use the solution
representation formula in Lemma \ref{lemma1} for (\ref{eq95}) to find the
bound on $|z_{i}|,$ and then the bound on $\tilde{\vartheta}_{i}%
:|\tilde{\vartheta}_{i}|\leq|z_{i}|+|G_{i}|.$ Comparatively to
\cite{zfo21cdc,zfo22tac}, this will greatly simplify the stability analysis
process along with the stability conditions, and improve the quantitative
bounds on the parameter $\varepsilon$ and time-delays $D_{i}$ $(i=1,\ldots,n)$.

Similar to Theorem \ref{theorem1}, we will find $k_{i}$ $(i=1,\ldots,n)$ from
the inequalities (\ref{eq23a}) with $\bar{D}_{i}=D_{i}$, which guarantee the
exponential stability with a decay rate $\delta_{i}$ in (\ref{eq23b}) of the
averaged system (\ref{eq23c}).

\begin{theorem}
\label{theorem2}Assume \textbf{A1-A3} hold. Let $k_{i}$ $(i=1,\ldots,n)$
satisfy (\ref{eq23a}) with $\bar{D}_{i}=D_{i}$. Given tuning parameters
$a_{i}$ $(i=1,\ldots,n)$ and $\kappa \geq0$ as well as $\bar{\sigma}_{i}%
>\bar{\sigma}_{0i}>0$ $(i=1,\ldots,n)$, let there exists $\varepsilon^{\ast
}>0$ that satisfy%
\begin{equation}
\left.
\begin{array}
[c]{l}%
\Phi_{1}^{i}=\left \vert k_{i}\bar{h}_{i}\right \vert D_{i\mathrm{M}%
}(\varepsilon^{\ast})-\frac{1}{\mathrm{e}}\leq0,\text{ }i=1,\ldots
,n\mathbf{,}\\
\Phi_{2}^{i}=\mathrm{e}^{\left \vert k_{i}\bar{h}_{i}\right \vert D_{i\mathrm{M}%
}(\varepsilon^{\ast})}\left[  \bar{\sigma}_{0i}+\frac{3\varepsilon^{\ast
}\left \vert k_{i}\right \vert \Delta_{1}}{2\left \vert a_{i}\right \vert }%
+\frac{W_{i}(\varepsilon^{\ast},\kappa)}{\left \vert k_{i}\bar{h}%
_{i}\right \vert }\right] \\
\text{ \  \  \  \ }+\frac{\varepsilon^{\ast}\left \vert k_{i}\right \vert
\Delta_{1}}{2\left \vert a_{i}\right \vert }-\bar{\sigma}_{i}<0,\text{
}i=1,\ldots,n\mathbf{,}%
\end{array}
\right.  \label{eq108}%
\end{equation}
where%
\[
\left.
\begin{array}
[c]{l}%
D_{i\mathrm{M}}=D_{i\mathrm{M}}(\varepsilon)=D_{i}+\frac{\varepsilon}{2^{n}%
},\text{ }i=1,\ldots,n\mathbf{,}\\
W_{i}(\varepsilon,\kappa)=\frac{\varepsilon \left \vert k_{i}\right \vert
}{\left \vert a_{i}\right \vert }\left(  \frac{\left \vert k_{i}\bar{h}%
_{i}\right \vert }{2}\Delta_{1}+\tilde{\Delta}_{1}+\tilde{\Delta}_{2}\right)
+\kappa \left \vert k_{i}\right \vert \sqrt{%
{\textstyle \sum \limits_{j=1}^{n}}
\bar{\sigma}_{j}^{2}},\\
\tilde{\Delta}_{1}=\frac{\left(  2^{n-1}+1\right)  \Delta_{1}}{2^{n}}\left(
{\textstyle \sum \limits_{j=1}^{n}}
\frac{\left \vert \bar{h}_{j}k_{j}\right \vert \bar{\sigma}_{j}}{\left \vert
a_{j}\right \vert }+\kappa \sqrt{%
{\textstyle \sum \limits_{j=1}^{n}}
\bar{\sigma}_{j}^{2}}\sqrt{%
{\textstyle \sum \limits_{j=1}^{n}}
\frac{k_{j}^{2}}{a_{j}^{2}}}\right)  ,\\
\tilde{\Delta}_{2}=\frac{\left(  2^{n-1}+1\right)  \Delta_{1}}{2^{n}}\left(
{\textstyle \sum \limits_{j=1}^{n}}
\left \vert \bar{h}_{j}k_{j}\right \vert +\kappa \sqrt{%
{\textstyle \sum \limits_{j=1}^{n}}
\bar{\sigma}_{j}^{2}}\sqrt{%
{\textstyle \sum \limits_{j=1}^{n}}
\frac{k_{j}^{2}}{a_{j}^{2}}}\right)
\end{array}
\right.
\]
with $\Delta_{1}$ given by (\ref{eq38}). Then for all $\varepsilon
\in(0,\varepsilon^{\ast}]$ satisfying (\ref{eq85}), the following holds:

\textbf{(i)} Solutions of (\ref{eq84a}) with $|\tilde{\vartheta}%
_{i}(D_{\mathrm{M}})|\leq \bar{\sigma}_{0i}$ satisfying the following bounds:%
\[
\left.
\begin{array}
[c]{l}%
|\tilde{\vartheta}_{i}(t)|<\left \vert \tilde{\vartheta}_{i}(D_{\mathrm{M}%
})\right \vert +\frac{\varepsilon \left \vert k_{i}\right \vert \Delta_{1}%
}{\left \vert a_{i}\right \vert }<\bar{\sigma}_{i},\text{ }t\in \lbrack
D_{\mathrm{M}},D_{\mathrm{M}}+\varepsilon],\\
|\tilde{\vartheta}_{i}(t)|<\left(  1+\left \vert k_{i}\bar{h}_{i}\right \vert
D_{i\mathrm{M}}\right)  \left[  \left \vert \tilde{\vartheta}_{i}%
(D_{\mathrm{M}})\right \vert +\frac{3\varepsilon \left \vert k_{i}\right \vert
\Delta_{1}}{2\left \vert a_{i}\right \vert }\right] \\
\text{ \  \  \  \  \  \  \  \  \ }+D_{i\mathrm{M}}W_{i}(\varepsilon,\kappa
)+\frac{\varepsilon \left \vert k_{i}\right \vert \Delta_{1}}{2\left \vert
a_{i}\right \vert }\\
\text{ \  \  \  \  \  \  \ }<\bar{\sigma}_{i},\text{ }t\in \lbrack D_{\mathrm{M}%
}+\varepsilon,D_{\mathrm{M}}+D_{i\mathrm{M}}+\varepsilon],\\
|\tilde{\vartheta}_{i}(t)|<\mathrm{e}^{-\left \vert k_{i}\bar{h}_{i}\right \vert
\left(  t-D_{\mathrm{M}}-2D_{i\mathrm{M}}-\varepsilon \right)  }\left[
\left \vert \tilde{\vartheta}_{i}(D_{\mathrm{M}})\right \vert +\frac
{3\varepsilon \left \vert k_{i}\right \vert \Delta_{1}}{2\left \vert
a_{i}\right \vert }\right] \\
\text{ \  \  \  \  \  \  \  \ }+\frac{\mathrm{e}^{\left \vert k_{i}\bar{h}%
_{i}\right \vert D_{i\mathrm{M}}}}{\left \vert k_{i}\bar{h}_{i}\right \vert
}W_{i}(\varepsilon,\kappa)+\frac{\varepsilon \left \vert k_{i}\right \vert
\Delta_{1}}{2\left \vert a_{i}\right \vert }\\
\text{ \  \  \  \  \  \ }<\bar{\sigma}_{i},\text{ }t\geq D_{\mathrm{M}%
}+D_{i\mathrm{M}}+\varepsilon.
\end{array}
\right.
\]
These solutions are exponentially attracted to the box%
\begin{equation}
\left.
\begin{array}
[c]{c}%
\vert \tilde{\vartheta}_{i}(t) \vert<\Theta_{i}\triangleq \frac{\mathrm{e}%
^{\left \vert k_{i}\bar{h}_{i}\right \vert D_{i\mathrm{M}}}}{\left \vert
k_{i}\bar{h}_{i}\right \vert }W_{i}(\varepsilon,\kappa)+\frac{\varepsilon
\left \vert k_{i}\right \vert \Delta_{1}}{2\left \vert a_{i}\right \vert },\\
i=1,\ldots,n
\end{array}
\right.  \label{eq74a}%
\end{equation}
with decay rates $\delta_{i}=\left \vert k_{i}\bar{h}_{i}\right \vert $, which
are independent of $\varepsilon$ and delays $D_{i}(t).$

\textbf{(ii)} Consider $\hat{\theta}(t)=U^{\mathrm{T}}\hat{\vartheta}(t)$ with
$U=[b_{ij}]_{n\times n},$ where $\hat{\vartheta}(t)$ is defined by
(\ref{eq55c}). Then the estimation errors $\tilde{\theta}_{i}(t)=\hat{\theta
}_{i}(t)-\theta^{\ast}$ such that $|\tilde{\theta}_{i}(D_{\mathrm{M}}%
)|\leq \sigma_{0i}$ with $%
{\textstyle \sum \nolimits_{i=1}^{n}}
\vert b_{ji} \vert \sigma_{0i}=\bar{\sigma}_{0j}$ $(j=1,\ldots,n\mathbf{)}$ are
exponentially attracted to the box $\vert \tilde{\theta}_{i}(t) \vert
<\sum_{j=1}^{n}b_{ji}\Theta_{j}$ $(i=1,\ldots,n)$ with a decay rate
$\delta=\min_{i=1,\ldots,n}\delta_{i}=\min_{i=1,\ldots,n}\left \vert k_{i}%
\bar{h}_{i}\right \vert ,$ where $\Theta_{j}$ is defined by (\ref{eq74a}).
\end{theorem}

\begin{proof}
The proof is similar to that of Theorem \ref{theorem1} and omitted.
\end{proof}

\begin{remark}
For the single input delay $D_{i}=D>0$ $(i=1,\ldots,n)$, by using (\ref{eq43})
we have
$
S(s_{p})=S(t-D), M(s_{p})=M(t-D), t\in \lbrack
t_{p},t_{p+1})
$
with $S(t)$ and $M(t)$ given by (\ref{eq93}). Note that the equations in
(\ref{eq85a}) still hold with $\tau$ replaced by $\tau-D.$ Thus, the result of
Theorem \ref{theorem2} holds for all $\varepsilon \in(0,\varepsilon^{\ast}],$
since there is no need to use (\ref{eq85}).
\end{remark}

\begin{remark}
In \cite{zfo22tac}, the single input delay $D$ was treated under restrictive
assumption that $D=$\textrm{O}$(\varepsilon)$ with $\varepsilon$ being a small
parameter. Comparatively to that, our results in Theorem \ref{theorem2} not
only allow distinct $D_{i},$ but also allow $D_{i}$ to be essentially larger
with arbitrary large constant $\bar D_i$. Actually, Theorem \ref{theorem2} guarantees
for any $D_{i}\geq0$ semi-global convergence for small enough $\varepsilon
^{\ast}$ and $\kappa.$ Moreover, Theorem \ref{theorem2} presents much simpler
stability conditions, which allow to get larger decay rate and period of the
dither signal in the numerical examples. Our UB on the estimation error in
(\ref{eq74a}) is of the order of $\mathrm{O}(\varepsilon)$ provided that
$\kappa=0$ and $a_{i},k_{i}$ $(i=1,\ldots,n)$ are of the order of
$\mathrm{O}(1).$ This is smaller than $\mathrm{O}(\sqrt{\varepsilon})$
achieved in \cite{zfo22tac}. In addition, \cite{zfo22tac} has no results for the uncertain $H$. As a comparison, by using the
presented time-delay approach, we can easily solve the uncertainty case.
\end{remark}

Parallel to Corollary \ref{corollary0}, we have:

\begin{corollary}
\label{corollary4} Assume that \textbf{A2-A4} hold. Given any $D_{i}>0$
$(i=1,\ldots,n)$ and $\bar{\sigma}_{0i}$ $>0$ (also $\sigma_{0i}>0$)
$(i=1,\ldots,n),$ and choosing $k_{i}$ $(i=1,\ldots,n)$ to satisfy
(\ref{eq23a}) with $\bar{D}_{i}=D_{i}$, the ES algorithm converges for small
enough $\varepsilon^{\ast}$ and $\kappa.$
\end{corollary}

\subsection{Single-variable static map}

Finally, we consider the single-variable static map (\ref{eq75b}) with
\textbf{A2'},\textbf{ A3}, \textbf{(}\ref{eq75}\textbf{) }and large known
input delay $D>0.$ Let the dither signals $S(t)$ and $M(t)$ satisfy%
\[
\left.  S(t)=a\cdot \mathrm{sq}\left(  \frac{t}{\varepsilon}\right)  ,\text{
}M(t)=\frac{1}{a}\cdot \mathrm{sq}\left(  \frac{t}{\varepsilon}\right)
\right.
\]
with $\varepsilon$ being a small parameter, and \{$s_{p}$\}, \{$t_{p}$\}
denote the sampling and control update instants, respectively, satisfying%
\[
\left.  s_{p}=\frac{p\varepsilon}{2},\text{ }t_{p}=s_{p}+D,\text{ }%
p\in \mathbf{Z}_{+}.\right.
\]
The ES algorithm is designed as%
\begin{equation}
\left.  \dot{\hat{\theta}}(t)=\left \{
\begin{array}
[c]{ll}%
0, & t\in \lbrack0,D),\\
kM(s_{p})y(s_{p}), & t\in \lbrack t_{p},t_{p+1})\text{ },p\in \mathbf{Z}_{+},
\end{array}
\right.  \right.  \label{eq830}%
\end{equation}
where $y(t)=Q(\theta(t))$ with $Q(\theta(t))$ given by (\ref{eq75b}) and
$\theta(t)=\hat{\theta}(t)+S(t).$ Denote%
\begin{equation}
\left.  \eta(t)=t-D-s_{p},\text{ }t\in \lbrack t_{p},t_{p+1}),\text{ }0\leq
\eta(t)<\frac{\varepsilon}{2^{n}}.\right.  \label{eq83a}%
\end{equation}
In addition, we can present%
\begin{equation}
\left.  M(s_{p})=M(t-D),\text{ }S(s_{p})=S(t-D),\text{ }t\in \lbrack
t_{p},t_{p+1}).\right.  \label{eq83b}%
\end{equation}
Taking into account (\ref{eq830})-(\ref{eq83b}), the estimation error
$\tilde{\theta}(t)=\hat{\theta}(t)-\theta^{\ast}$ is governed by%
\begin{equation}
\left.
\begin{array}
[c]{l}%
\dot{\tilde{\theta}}(t)=kM(t-D)Q^{\ast}+\frac{1}{2}khM(t-D)\\
\text{ \  \  \  \  \  \  \  \ }\times \lbrack S(t-D)+\tilde{\theta}(t-D(t))]^{2}%
,\text{ }t\geq D,\\
\tilde{\theta}(t)=\tilde{\theta}(D),\text{ }t\in \lbrack0,D].
\end{array}
\right.  \label{eq76b}%
\end{equation}
Applying the time-delay approach to averaging of (\ref{eq76b}), we finally
arrive at%
\[
\left.
\begin{array}
[c]{l}%
\dot{z}(t)=khz(t-D(t))+\bar{w}(t),\text{ }t\geq D+\varepsilon,\text{ }%
i\in \mathbf{I[}1,n\mathbf{]},\\
\bar{w}(t)=khG(t-D(t))-Y_{1}(t)-Y_{2}(t)
\end{array}
\right.
\]
with $\{z(t),G(t),Y_{1}(t),Y_{2}(t)\}$ satisfying updated (\ref{eq52a}),
(\ref{eq20}) and (\ref{eq20a}), respectively.

Similar to Corollary \ref{corollary2}, we will find $k$ from the inequality
(\ref{eq78}) with $D^{\mathrm{in}}+D=D$, which guarantees the exponential
stability with a decay rate $\delta$ in (\ref{eq78a}) of the averaged system
(\ref{eq78b}).

\begin{corollary}
\label{corollary3}Assume that \textbf{A2'},\textbf{ A3} and \textbf{(}%
\ref{eq75}\textbf{) }hold\textbf{. }let $k$ satisfy (\ref{eq78}) with
$D^{\mathrm{in}}+D=D$. Given tuning parameters $a$ and $\sigma>\sigma_{0}>0$,
let there exists $\varepsilon^{\ast}>0$ that satisfy%
\[
\left.
\begin{array}
[c]{l}%
\Phi_{1}=\left \vert k\right \vert h_{\mathrm{M}}\left(  D+\frac{\varepsilon
^{\ast}}{2}\right)  -\frac{1}{\mathrm{e}}\leq0\mathbf{,}\\
\Phi_{2}=\mathrm{e}^{\left \vert k\right \vert h_{\mathrm{M}}\left(
D+\varepsilon^{\ast}/2\right)  }\left[  \sigma_{0}+\frac{3\varepsilon^{\ast
}\Delta}{4}+W(\varepsilon^{\ast})\right] \\
\text{ \  \  \  \ }+\frac{\varepsilon^{\ast}\Delta}{4}-\sigma<0\mathbf{,}%
\end{array}
\right.
\]
where
\[
\left.  W(\varepsilon)=\frac{\varepsilon \Delta}{4}\left(  3+\frac{2\sigma
}{\left \vert a\right \vert }\right)  \right.
\]
and $\Delta$ is given in (\ref{eq109}). Then for all $\varepsilon
\in(0,\varepsilon^{\ast}],$ the solution $\tilde{\theta}(t)$ with
$|\tilde{\theta}(D)|\leq \sigma_{0}$ satisfies%
\[
\left.
\begin{array}
[c]{l}%
|\tilde{\theta}(t)|<\left \vert \tilde{\theta}(D)\right \vert +\frac
{\varepsilon}{2}\Delta<\sigma,\text{ }t\in \lbrack D,D+\varepsilon],\\
|\tilde{\theta}(t)|<\left[  1+\left \vert k\right \vert h_{\mathrm{M}}\left(
D+\frac{\varepsilon}{2}\right)  \right]  \left[  \left \vert \tilde{\theta
}(D)\right \vert +\frac{3\varepsilon \Delta}{4}\right] \\
\text{ \  \  \  \  \  \  \  \  \ }+\left(  D+\frac{\varepsilon}{2}\right)  \left \vert
k\right \vert h_{\mathrm{M}}W(\varepsilon)+\frac{\varepsilon \Delta}{4}\\
\text{ \  \  \  \  \  \  \ }<\sigma,\text{ }t\in \left[  D+\varepsilon,2D+\frac
{3\varepsilon}{2}\right]  .\\
|\tilde{\theta}(t)|<\mathrm{e}^{-\left \vert k\right \vert h_{\mathrm{m}}\left(
t-2D-3\varepsilon/2\right)  }\mathrm{e}^{\left \vert k\right \vert
h_{\mathrm{M}}\left(  D+\varepsilon/2\right)  }\left[  \left \vert
\tilde{\vartheta}(D)\right \vert +\frac{3\varepsilon \Delta}{4}\right] \\
\text{ \  \  \  \  \  \  \  \  \ }+\mathrm{e}^{\left \vert k\right \vert h_{\mathrm{M}%
}\left(  D+\varepsilon/2\right)  }W(\varepsilon)+\frac{\varepsilon \Delta}{4}\\
\text{ \  \  \  \  \  \  \ }<\sigma,\text{ }t\geq2D+\frac{3\varepsilon}{2}.
\end{array}
\right.
\]
This solution is exponentially attracted to the interval%
\[
\left.  |\tilde{\theta}(t)|<\mathrm{e}^{\left \vert k\right \vert h_{\mathrm{M}%
}\left(  D+\varepsilon/2\right)  }W(\varepsilon)+\frac{\varepsilon \Delta}%
{4}\right.
\]
with a decay rate $\delta=\left \vert k\right \vert h_{\mathrm{m}}.$
\end{corollary}

\subsection{Examples\label{sec4.2}}

Example 4.1. ($1\mathrm{D}$\textit{ static map}) Consider the single-input map
(\ref{eq75b}) with%
\begin{equation}
Q^{\ast}=0,\text{ }h=2.\label{eq113}%
\end{equation}
Note that%
\[
\left.
\begin{array}
[c]{ll}%
\dot{\hat{\theta}}(t)=ka\cdot \mathrm{sq}(\frac{t}{\varepsilon})y(t) & \text{in
\cite{zfo21cdc},}\\
\dot{\hat{\theta}}(t)=\frac{k}{a}\cdot \mathrm{sq}(\frac{t}{\varepsilon})y(t) &
\text{in Corollary \ref{corollary3}.}%
\end{array}
\right.
\]
For a fair comparison, we select the tuning parameters of the gradient-based
ES as%
\begin{equation}
\left.
\begin{array}
[c]{ll}%
a=0.1,\text{ }k=-1.3 & \text{in \cite{zfo21cdc},}\\
a=0.1,\text{ }k=-1.3\cdot0.1^{2}=-0.013 & \text{in Corollary \ref{corollary3}%
.}%
\end{array}
\right.  \label{eq114a}%
\end{equation}
If $Q^{\ast}$ and $h$ are uncertain and satisfy \textbf{A3} and (\ref{eq75}),
respectively, we consider%
\begin{align}
\left \vert Q^{\ast}\right \vert  &  \leq0.1,\text{ }1.9\leq \left \vert
h\right \vert \leq2.1;\label{eq114}\\
\left \vert Q^{\ast}\right \vert  &  \leq1.0,\text{ }1.0\leq \left \vert
h\right \vert \leq3.0.\label{eq115}%
\end{align}
The results that follow from Theorem \ref{theorem2}, Corollary
\ref{corollary3} and \cite{zfo21cdc} are shown in Table \ref{table4}. It
follows that our results allow larger decay rate $\delta$, larger upper bound
$\varepsilon^{\ast}$ (small lower bound frequency $\omega^{\ast}$), much
larger time-delay $D$ and much smaller UB than those in \cite{zfo21cdc}.
Moreover, our results allow much larger uncertainties in  $Q^{\ast}$ and $h$
than those in \cite{zfo21cdc}.

\begin{table}[h]
\caption{Example 4.1: Comparison of $\delta$, $D$, $\varepsilon^{\ast}$, $\omega^{\ast
}=2\pi/\varepsilon^{\ast}$ and UB}
\label{table4}%
\center {\tiny
\setlength{\tabcolsep}{3pt}
\begin{tabular}
[c]{|l|c|c|c|c|c|c|c|}\hline
ES: square & $\sigma_{0}$ & $\sigma$ & $\delta$ & $D$ & $\varepsilon^{\ast}$ &$\omega
^{\ast}$ & UB\\ \hline
Corollary \ref{corollary3} with (\ref{eq113}) & 1 & $\sqrt{2}$ & 0.026 & 1.0 &
0.071 &88.49 & $2.7\cdot \mathrm{10}^{-4}$\\ \hline
Theorem \ref{theorem2} with (\ref{eq113}) & 1 & $\sqrt{2}$ & 0.026 & 1.0 &
0.071 &88.49 & $2.7\cdot \mathrm{10}^{-4}$\\ \hline
\cite{zfo21cdc} with (\ref{eq113}) & 1 & $\sqrt{2}$ & 0.02 & 0.01 & 0.045 &139.63 &
0.04\\ \hline
Corollary \ref{corollary3} with (\ref{eq114}) & 1 & $\sqrt{2}$ & 0.0247 &
1.0 & 0.065 &96.66 & $2.0\cdot \mathrm{10}^{-3}$\\ \hline
Theorem \ref{theorem2} with (\ref{eq114}) & 1 & $\sqrt{2}$ & 0.026 & 0.5 &
0.052 &120.83 & $1.9\cdot \mathrm{10}^{-3}$\\ \hline
\cite{zfo21cdc} with (\ref{eq114}) & 1 & $\sqrt{2}$ & 0.02 & 0.01 & 0.036 &174.53 &
0.22\\ \hline
Corollary \ref{corollary3} with (\ref{eq115}) & 1 & $\sqrt{2}$ & 0.013 & 1.0 &
0.035 &179.52 & $1.1\cdot \mathrm{10}^{-2}$\\ \hline
Theorem \ref{theorem2} with (\ref{eq115}) & 1 & $\sqrt{2}$ & - & - & - & - &
-\\ \hline
\cite{zfo21cdc} with (\ref{eq115}) & 1 & $\sqrt{2}$ & - & - & - & -  & -\\ \hline
\end{tabular} }\end{table}

Example 4.2. ($2\mathrm{D}$\textit{ static map}) We consider the Example 3.2 in
Section \ref{sec2.2} with $D_{i}=D$ ($i=1,2$) and the parameters chosen as
\cite{zfo22tac}:
\begin{equation}
k_{1}=k_{2}=-0.01, { \ } a_{1}=a_{2}=0.2. \label{eq2aa}
\end{equation}
Verifying conditions
of Theorem \ref{theorem2} and \cite{zfo22tac}, we arrive at results shown in
Table \ref{table5}, in which UB$=($\textrm{\={B}}$_{1}^{2}+$\textrm{\={B}%
}$_{2}^{2})^{1/2}$ with \textrm{\={B}}$_{1}$ and \textrm{\={B}}$_{2}$ be the
ultimate bound values of $\tilde{x}(t)$ and $\tilde{y}(t)$, respectively. By
comparing the data, we find that our results in Theorem \ref{theorem2} allow
larger decay rate $\delta$, larger upper bound $\varepsilon^{\ast}$ (small
lower bound frequency $\omega^{\ast}$), much larger time-delay $D$ and much
smaller UB than those in \cite{zfo22tac}.

\begin{table}[h]
\caption{Example 4.2 with (\ref{eq2aa}): Comparison of $\delta$, $D$, $\varepsilon^{\ast}$, $\omega^{\ast
}=2\pi/\varepsilon^{\ast}$ and UB}
\label{table5}%
\center  {\tiny { \setlength{\tabcolsep}{3pt}
\begin{tabular}
[c]{|l|c|c|c|c|c|c|c|}\hline
ES: square & $\bar{\sigma}_{01}/\bar{\sigma}_{02}$ & $\bar{\sigma}_{1}/\bar{\sigma}_{2}$  & $\delta$ & $D$ & $\varepsilon^{\ast}$ & $\omega^{\ast}$ & UB \\ \hline
Theorem 2 & 1/1 & 2/2  & $0.02$ &$2$  &0.1 &62.83 & $1.6\cdot \mathrm{10}^{-3}$\\ \hline
\cite{zfo22tac} & 1/1 & 2/2  & $0.01$ &$0.01$  &0.09 &69.81& 0.18\\ \hline
\end{tabular}
} }\end{table}

Example 4.3 ($2\mathrm{D}$\textit{ static map}) We consider the Example 3.3 in Section
\ref{sec2.2} with the delays and tuning parameters chosen as%
\begin{equation}
\left.
\begin{array}
[c]{l}%
D_{1}=1.5,\text{ }D_{2}=2.5,\text{ }a_{1}=a_{2}=0.3,\\
k_{1}=0.6\cdot10^{-3},\text{ }k_{2}=0.4\cdot10^{-2}.
\end{array}
\right.  \label{eq27}%
\end{equation}
Following (\ref{eq85}), we choose $\varepsilon$ as in (\ref{eq60a}). The
results that follow from Theorem \ref{theorem2} are shown in Table
\ref{table6}, in which UB ($\mathrm{\bar{B}}_{i},$ $i=1,2$) corresponds to
$\varepsilon=1/2$ $(<\varepsilon^{\ast}=0.79)$ and $\varepsilon=1/4$
$(<\varepsilon^{\ast}=0.36)$ for known ($\kappa=0$) and uncertain ($\kappa=0.2$) $H,$ respectively. It
follows that Theorem \ref{theorem2} works well.

\begin{table}[h]
\caption{Example 4.3 with (\ref{eq27}): The values of $\varepsilon^{\ast}$, $\omega^{\ast}(2\pi
/\varepsilon^{\ast})$, $\delta_{i}$ (i=1,2) and UB ($\mathrm{\bar{B}}_{i},$
$i=1,2$)}
\label{table6}
\center  {\tiny { \setlength{\tabcolsep}{3pt}
\begin{tabular}
[c]{|l|c|c|c|c|c|c|}\hline
ES: square & $\bar{\sigma}_{01}/\bar{\sigma}_{02}$ & $\bar{\sigma}_{1}/\bar{\sigma}_{2}$ & $\varepsilon^{\ast}$ & $\omega
^{\ast}$ & $\delta_{1}=\delta_{2}$ & $\mathrm{\bar{B}}_{1}/\mathrm{\bar{B}}_{2}$\\ \hline
Known $H$ & 0.5/0.5 & 0.6/1.0 &0.79 &7.95& $0.031$ & $4.1\cdot \mathrm{10}^{-3}$/$2.8\cdot \mathrm{10}^{-2}$\\ \hline
Uncertain $H$ & 0.5/0.5 & 0.6/1.0 &0.36 &17.45 & $0.031$ &$5.0\cdot \mathrm{10}^{-3}$/$3.4\cdot \mathrm{10}^{-2}$\\ \hline
\end{tabular}
} }\end{table}For the real-time estimate $\hat{\theta}(t)$ with the estimation
error $\tilde{\theta}(t),$ following part (ii) of Theorem \ref{theorem2}, we
solve (\ref{eq16a}) to obtain%
\[
\left.  \sigma_{01}=\sigma_{02}=0.3633,\right.
\]
and the UB ($\mathrm{B}_{i}$ for $\tilde{\theta}_{i}(t),$ $i=1,2$) can be
calculated as%
\[
\left.
\begin{array}
[c]{l}%
\mathrm{B}_{1}=0.5257\cdot \mathrm{\bar{B}}_{1}+0.8507\cdot \mathrm{\bar{B}}%
_{2}=2.6\cdot10^{-2},\\
\mathrm{B}_{2}=0.8507\cdot \mathrm{\bar{B}}_{1}+0.5257\cdot \mathrm{\bar{B}}%
_{2}=1.82\cdot10^{-2}.
\end{array}
\right.
\]

\section{Conclusion\label{sec5}}

This paper developed a time-delay approach to gradient-based ES of uncertain
quadratic maps in the presence of large measurement delay and input distinct
delays, and sampled-data ES with large distinct input delays. By choosing
gains to stabilize averaged time-delay systems, explicit conditions in terms
of simple inequalities were established to guarantee the practical stability
of the ES control systems. The resulting time-delay method provides a
quantitative bounds on the control parameters and the ultimate bound of
seeking error. Compared with recent quantitative result on the sampled-data
implementation, the presented method not only greatly simplifies the
stability conditions and improves the results, but also allow distinct and
large time delays. Future work may include the extension to discrete-time ES
in the presence of large delays by using the time-delay approach to averaging for
discrete-time systems introduced in \cite{yzf22tac,yf22tac}.

\section*{Appendix}

\subsection*{A1: Proof of Lemma \ref{lemma2}}

Let%
\begin{equation}
\left.  i\cdot \omega_{j}=j\cdot \omega_{i},\text{ }\forall i,j=1,\ldots
,n\mathbf{.}\right.  \label{eq61}%
\end{equation}
For fixed $m_{1},$ we choose%
\begin{equation}
\left.  \omega_{1}=\frac{q\cdot2\pi m_{1}}{\bar{D}_{1}},\text{ }q\in
\mathbf{N}.\right.  \label{eq62}%
\end{equation}
Then via (\ref{eq61}), $\omega_{i}$ ($i=2,\ldots,n$) are in forms of%
\begin{equation}
\left.  \omega_{i}=\frac{q\cdot2\pi im_{1}}{\bar{D}_{1}},\text{ }%
q\in \mathbf{N},\text{ }\forall i=2,\ldots,n\mathbf{.}\right.  \label{eq63}%
\end{equation}
Combining (\ref{eq62}) and (\ref{eq63}) gives
\begin{equation}
\left.  \omega_{i}=\frac{q\cdot2\pi im_{1}}{\bar{D}_{1}},\text{ }%
q\in \mathbf{N},\text{ }\forall i=1,\ldots,n\mathbf{.}\right.  \label{eq63a}%
\end{equation}
Noting (\ref{eq61a}), from (\ref{eq63a}) we further have (\ref{eq66a}),
namely,%
\[
\left.  \omega_{i}\bar{D}_{j}=2\pi \cdot qim_{j},\text{ }q\in \mathbf{N},\text{
}\forall i,j=1,\ldots,n\mathbf{,}\right.
\]
which implies (\ref{eq66}) since $qim_{j}$ are all positive integers. This
completes the proof.

\subsection*{A2: Proof of Theorem \ref{theorem1}}

Since $\tilde{\theta}(t)=U^{\mathrm{T}}\tilde{\vartheta}(t),$ it is not
difficult to obtain part (ii) from part (i). Thus, we just need to prove part
(i). The proof is divided into three parts. (A) First, we present a group of
upper bounds under the assumption that $\tilde{\vartheta}_{i}(t)$
($i=1,\ldots,n$) are bounded for $t\geq D_{\mathrm{M}}$; (B) Second, we show
the practical stability of each $z_{i}$-system in (\ref{eq21}) (and thus each
$\tilde{\vartheta}_{i}$-system in (\ref{eq112})); (C) Third, we show the
availability of the assumption that $\tilde{\vartheta}_{i}(t)$ are bounded for
$t\geq D_{\mathrm{M}}$ by contradiction.

\textit{Proof of part A.} Assume that
\begin{equation}
\left.  \vert \tilde{\vartheta}_{i}(t) \vert<\bar{\sigma}_{i},\text{
}i=1,\ldots,n,\text{ }t\geq D_{\mathrm{M}}.\right.  \label{eq25}%
\end{equation}
When $t\in \lbrack0,D_{\mathrm{M}}],$ we note that $\tilde{\vartheta}_{i}(t)$
($i=1,\ldots,n$) are constants satisfying
\begin{equation}
\left.  \vert \dot{\tilde{\vartheta}}_{i}(t) \vert=0,\text{ } \vert
\tilde{\vartheta}_{i}(t) \vert \leq \bar{\sigma}_{0i}<\bar{\sigma}_{i},\text{
}t\in \lbrack0,D_{\mathrm{M}}],\right.  \label{eq39}%
\end{equation}
which with (\ref{eq25}) yields
\begin{equation}
\left.  \vert \tilde{\vartheta}_{i}(t) \vert<\bar{\sigma}_{i},\text{ }%
t\geq0,\text{ }i=1,\ldots,n\mathbf{,}\right.  \label{eq25a}%
\end{equation}
and
\begin{equation}
\left.  \vert \tilde{\vartheta}(t) \vert=\sqrt{%
{\textstyle \sum \limits_{j=1}^{n}}
\tilde{\vartheta}_{j}^{2}(t)}<\sqrt{%
{\textstyle \sum \limits_{j=1}^{n}}
\bar{\sigma}_{j}^{2}},\text{ }t\geq0.\right.  \label{eq42}%
\end{equation}
Employing (\ref{eq26}) and (\ref{eq20c}), we find%
\begin{equation}
\left.
\begin{array}
[c]{l}%
\text{ \  \ }Q(U^{\mathrm{T}}\vartheta(t-D_{i}(t)))\\
\overset{(\ref{eq11})}{=}Q^{\ast}+\frac{1}{2}%
{\textstyle \sum \limits_{j=1}^{n}}
\bar{h}_{j}\left(  \tilde{\vartheta}_{j}(t-D_{i}(t))+S_{j}(t-D_{i}(t))\right)
^{2}\\
\text{ \  \ }+\frac{1}{2}\left[  \tilde{\vartheta}(t-D_{i}(t))+S(t-D_{i}%
(t))\right]  ^{\mathrm{T}}\\
\text{ \  \ }\times \Delta \mathcal{H}\left[  \tilde{\vartheta}(t-D_{i}%
(t))+S(t-D_{i}(t))\right]  ,\\
\text{ \  \ }\bar{Q}(\tilde{\vartheta}(t-D_{i}(t)))\\
\overset{(\ref{eq11})}{=}Q^{\ast}+\frac{1}{2}%
{\textstyle \sum \limits_{j=1}^{n}}
\bar{h}_{j}\left(  \tilde{\vartheta}_{j}(t-D_{i}(t))+S_{j}(t)\right)  ^{2}\\
\text{ \  \ }+\frac{1}{2}\left[  \tilde{\vartheta}(t-D_{i}(t))+S(t)\right]
^{\mathrm{T}}\\
\text{ \  \ }\times \Delta \mathcal{H}\left[  \tilde{\vartheta}(t-D_{i}%
(t))+S(t)\right]  .
\end{array}
\right.  \label{eq44b}%
\end{equation}
Note from (\ref{eq32}), (\ref{eq55}) and $\left \Vert \Delta H(t)\right \Vert
\leq \kappa$ in \textbf{A3} that
\begin{equation}
\left.
\begin{array}
[c]{l}%
\left \vert S_{i}(t)\right \vert =\left \vert a_{i}\right \vert ,\text{
}\left \vert M_{i}(t)\right \vert =\frac{2}{\left \vert a_{i}\right \vert },\text{
}i=1,\ldots,n\mathbf{,}\\
\left \vert S^{\mathrm{T}}(t)\right \vert =\left \vert S(t)\right \vert =\sqrt{%
{\textstyle \sum \limits_{j=1}^{n}}
a_{j}^{2}},\\
\left \Vert \Delta \mathcal{H}(t)\right \Vert \leq \left \Vert U\right \Vert
^{2}\left \Vert \Delta H(t)\right \Vert \leq \kappa.
\end{array}
\right.  \label{eq44}%
\end{equation}
Then via (\ref{eq25a})-(\ref{eq44}), we have
\begin{equation}
\left.
\begin{array}
[c]{l}%
\text{ \  \ }\left \vert Q(U^{\mathrm{T}}\vartheta(t-D_{i}(t)))\right \vert \\
\leq \left \vert Q^{\ast}\right \vert +\frac{1}{2}%
{\textstyle \sum \limits_{j=1}^{n}}
\left \vert \bar{h}_{j}\right \vert \vert \tilde{\vartheta}_{j}(t-D_{i}%
(t))+S_{j}(t-D_{i}(t)) \vert^{2}\\
\text{ \  \ }+\frac{1}{2}\left \Vert \Delta \mathcal{H}\right \Vert \vert
\tilde{\vartheta}(t-D_{i}(t))+S(t-D_{i}(t)) \vert^{2}\\
<\Delta_{1},\text{ }t\geq D_{\mathrm{M}},\\
\text{ \  \ } \vert \bar{Q}(\tilde{\vartheta}(t-D_{i}(t))) \vert \\
\leq \left \vert Q^{\ast}\right \vert +\frac{1}{2}%
{\textstyle \sum \limits_{j=1}^{n}}
\left \vert \bar{h}_{j}\right \vert \vert \tilde{\vartheta}_{j}(t-D_{i}%
(t))+S_{j}(t) \vert^{2}\\
\text{ \  \ }+\frac{1}{2}\left \Vert \Delta \mathcal{H}\right \Vert \vert
\tilde{\vartheta}(t-D_{i}(t))+S(t) \vert^{2}\\
<\Delta_{1},\text{ }t\geq D_{\mathrm{M}}%
\end{array}
\right.  \label{eq28}%
\end{equation}
with $\Delta_{1}$ given in (\ref{eq38}), by which, (\ref{eq39}) and
(\ref{eq44}) we further have for $i=1,\ldots,n\mathbf{,}$
\begin{equation}
\left.  \left \vert \dot{\tilde{\vartheta}}_{i}(t)\right \vert <\frac
{2\left \vert k_{i}\right \vert \Delta_{1}}{\left \vert a_{i}\right \vert },\text{
}t\geq0,\right.  \label{eq28a}%
\end{equation}
and then
\begin{equation}
\left.
\begin{array}
[c]{l}%
\left \vert \tilde{\vartheta}_{i}(t)\right \vert =\left \vert \tilde{\vartheta
}_{i}(D_{\mathrm{M}})+%
{\textstyle \int \nolimits_{D_{\mathrm{M}}}^{t}}
\dot{\tilde{\vartheta}}_{i}(s)\mathrm{d}s\right \vert \\
<\left \vert \tilde{\vartheta}_{i}(D_{\mathrm{M}})\right \vert +\frac
{2\varepsilon \left \vert k_{i}\right \vert \Delta_{1}}{\left \vert a_{i}%
\right \vert },\text{ }t\in \lbrack D_{\mathrm{M}},D_{\mathrm{M}}+\varepsilon].
\end{array}
\right.  \label{eq64}%
\end{equation}
This implies the first inequality in (\ref{eq65a}) since $\Phi_{2}^{i}<0$ in
(\ref{eq24a}) implies that $\bar{\sigma}_{0i}+\frac{2\varepsilon^{\ast
}\left \vert k_{i}\right \vert \Delta_{1}}{\left \vert a_{i}\right \vert }%
<\bar{\sigma}_{i}.$ Moreover, employing (\ref{eq28a}) we have
\begin{equation}
\left.  \left \vert \dot{\tilde{\vartheta}}(t)\right \vert =\sqrt{%
{\textstyle \sum \limits_{j=1}^{n}}
\dot{\tilde{\vartheta}}_{j}(t)}<2\Delta_{1}\sqrt{%
{\textstyle \sum \limits_{j=1}^{n}}
\frac{k_{j}^{2}}{a_{j}^{2}}},\text{ }t\geq0.\right.  \label{eq64a}%
\end{equation}
By using (\ref{eq44}) and the second inequality in (\ref{eq28}), from
(\ref{eq20}) we have
\begin{equation}
\left.
\begin{array}
[c]{l}%
\left \vert G_{i}(t)\right \vert =\left \vert \frac{1}{\varepsilon}%
\int \nolimits_{t-\varepsilon}^{t}(\tau-t+\varepsilon)k_{i}M_{i}(\tau)\bar
{Q}(\tilde{\vartheta}(\tau-D_{i}(\tau)))\mathrm{d}\tau \right \vert \\
\text{ \  \  \  \  \  \  \ }\leq \frac{1}{\varepsilon}%
{\textstyle \int \nolimits_{t-\varepsilon}^{t}}
(\tau-t+\varepsilon)\left \vert k_{i}M_{i}(\tau)\bar{Q}(\tilde{\vartheta}%
(\tau-D_{i}(\tau)))\right \vert \mathrm{d}\tau \\
\text{ \  \  \  \  \  \  \ }<\frac{2\left \vert k_{i}\right \vert \Delta_{1}%
}{\varepsilon \left \vert a_{i}\right \vert }%
{\textstyle \int \nolimits_{t-\varepsilon}^{t}}
(\tau-t+\varepsilon)\mathrm{d}\tau \\
\text{ \  \  \  \  \  \  \ }=\frac{\varepsilon \left \vert k_{i}\right \vert \Delta
_{1}}{\left \vert a_{i}\right \vert },\text{ }t\geq D_{\mathrm{M}}%
+\varepsilon,\text{ }i=1,\ldots,n\mathbf{,}%
\end{array}
\right.  \label{eq13}%
\end{equation}
this with $G_{i}(t)=0,$ $t\in \lbrack \varepsilon,D_{\mathrm{M}}+\varepsilon)$
in (\ref{eq52a}) gives \vspace{-0.2cm}
\begin{equation}
\left.
\begin{array}
[c]{l}%
|k_{i}\bar{h}_{i}G_{i}(t-D_{i}(t))|\leq \left \vert k_{i}\bar{h}_{i}\right \vert
|G_{i}(t-D_{i}(t))|\\
<\frac{\varepsilon k_{i}^{2}\left \vert \bar{h}_{i}\right \vert \Delta_{1}%
}{\left \vert a_{i}\right \vert },\text{ }t\geq D_{\mathrm{M}}+\varepsilon
,\text{ }i=1,\ldots,n\mathbf{.}%
\end{array}
\right.  \label{eq45}%
\end{equation}
Employing (\ref{eq11}), (\ref{eq21a}), (\ref{eq25a}), (\ref{eq42}),
(\ref{eq44}), (\ref{eq28a}) and (\ref{eq64a}), we have from (\ref{eq20a}) that%
\begin{equation}
\left.
\begin{array}
[c]{l}%
|Y_{1i}(t)|\leq \frac{1}{\varepsilon}%
{\textstyle \int \nolimits_{t-\varepsilon}^{t}}
{\textstyle \int \nolimits_{\tau-D_{i}(\tau)}^{t-D_{i}(t)}}
\left \vert k_{i}M_{i}(\tau)\right \vert \left \vert \tilde{\vartheta
}^{\mathrm{T}}(s)\mathcal{H}\dot{\tilde{\vartheta}}(s)\right \vert
\mathrm{d}s\mathrm{d}\tau \\
\text{ \  \  \  \  \  \  \  \ }\leq \frac{2\left \vert k_{i}\right \vert }%
{\varepsilon \left \vert a_{i}\right \vert }%
{\textstyle \int \nolimits_{t-\varepsilon}^{t}}
{\textstyle \int \nolimits_{\tau-D_{i}(\tau)}^{t-D_{i}(t)}}
{\textstyle \sum \limits_{j=1}^{n}}
\left \vert \bar{h}_{j}\tilde{\vartheta}_{j}(s)\dot{\tilde{\vartheta}}%
_{j}(s)\right \vert \mathrm{d}s\mathrm{d}\tau \\
\text{ \  \  \  \  \  \  \  \  \  \  \ }+\frac{2\left \vert k_{i}\right \vert
}{\varepsilon \left \vert a_{i}\right \vert }%
{\textstyle \int \nolimits_{t-\varepsilon}^{t}}
{\textstyle \int \nolimits_{\tau-D_{i}(\tau)}^{t-D_{i}(t)}}
\left \vert \tilde{\vartheta}^{\mathrm{T}}(s)\Delta \mathcal{H}\dot
{\tilde{\vartheta}}(s)\right \vert \mathrm{d}s\mathrm{d}\tau \\
\text{ \  \  \  \  \  \  \  \ }<\frac{4\left \vert k_{i}\right \vert \Delta_{1}%
}{\varepsilon \left \vert a_{i}\right \vert }%
{\textstyle \sum \limits_{j=1}^{n}}
\frac{\left \vert \bar{h}_{j}k_{j}\right \vert \bar{\sigma}_{j}}{\left \vert
a_{j}\right \vert }%
{\textstyle \int \nolimits_{t-\varepsilon}^{t}}
{\textstyle \int \nolimits_{\tau-D_{i}(\tau)}^{t-D_{i}(t)}}
1\mathrm{d}s\mathrm{d}\tau \\
\text{ \  \  \  \  \  \  \  \  \  \  \ }+\frac{4\kappa \left \vert k_{i}\right \vert
\Delta_{1}}{\varepsilon \left \vert a_{i}\right \vert }\sqrt{%
{\textstyle \sum \limits_{j=1}^{n}}
\bar{\sigma}_{j}^{2}}\sqrt{%
{\textstyle \sum \limits_{j=1}^{n}}
\frac{k_{j}^{2}}{a_{j}^{2}}}%
{\textstyle \int \nolimits_{t-\varepsilon}^{t}}
{\textstyle \int \nolimits_{\tau-D_{i}(\tau)}^{t-D_{i}(t)}}
1\mathrm{d}s\mathrm{d}\tau \\
\text{ \  \  \  \  \  \  \  \ }\leq \frac{4\left \vert k_{i}\right \vert \Delta_{1}%
}{\varepsilon \left \vert a_{i}\right \vert }\left(
{\textstyle \sum \limits_{j=1}^{n}}
\frac{\left \vert \bar{h}_{j}k_{j}\right \vert \bar{\sigma}_{j}}{\left \vert
a_{j}\right \vert }+\kappa \sqrt{%
{\textstyle \sum \limits_{j=1}^{n}}
\bar{\sigma}_{j}^{2}}\sqrt{%
{\textstyle \sum \limits_{j=1}^{n}}
\frac{k_{j}^{2}}{a_{j}^{2}}}\right) \\
\text{ \  \  \  \  \  \  \  \  \  \  \ }\times%
{\textstyle \int \nolimits_{t-\varepsilon}^{t}}
\left(  t-\tau+2\varepsilon \mu \right)  \mathrm{d}\tau \\
\text{ \  \  \  \  \  \  \  \ }=\frac{\varepsilon \left \vert k_{i}\right \vert
}{\left \vert a_{i}\right \vert }\bar{\Delta}_{1},\text{ }t\geq D_{\mathrm{M}%
}+\varepsilon,\text{ }i=1,\ldots,n
\end{array}
\right.  \label{eq14}%
\end{equation}
with $\bar{\Delta}_{1}$ given in (\ref{eq38}),%
\begin{equation}
\left.
\begin{array}
[c]{l}%
|Y_{2i}(t)|\leq \frac{1}{\varepsilon}%
{\textstyle \int \nolimits_{t-\varepsilon}^{t}}
{\textstyle \int \nolimits_{\tau-D_{i}(\tau)}^{t-D_{i}(t)}}
\left \vert k_{i}M_{i}(\tau)\right \vert \left \vert S^{\mathrm{T}}%
(\tau)\mathcal{H}\dot{\tilde{\vartheta}}(s)\right \vert \mathrm{d}%
s\mathrm{d}\tau \\
\text{ \  \  \  \  \  \  \  \ }\leq \frac{2\left \vert k_{i}\right \vert }%
{\varepsilon \left \vert a_{i}\right \vert }%
{\textstyle \int \nolimits_{t-\varepsilon}^{t}}
{\textstyle \int \nolimits_{\tau-D_{i}(\tau)}^{t-D_{i}(t)}}
{\textstyle \sum \limits_{j=1}^{n}}
\left \vert \bar{h}_{j}S_{j}(\tau)\dot{\tilde{\vartheta}}_{j}(s)\right \vert
\mathrm{d}s\mathrm{d}\tau \\
\text{ \  \  \  \  \  \  \  \  \  \  \ }+\frac{2\left \vert k_{i}\right \vert
}{\varepsilon \left \vert a_{i}\right \vert }%
{\textstyle \int \nolimits_{t-\varepsilon}^{t}}
{\textstyle \int \nolimits_{\tau-D_{i}(\tau)}^{t-D_{i}(t)}}
\left \vert S^{\mathrm{T}}(\tau)\Delta \mathcal{H}\dot{\tilde{\vartheta}%
}(s)\right \vert \mathrm{d}s\mathrm{d}\tau \\
\text{ \  \  \  \  \  \  \  \ }<\frac{4\left \vert k_{i}\right \vert \Delta_{1}%
}{\varepsilon \left \vert a_{i}\right \vert }%
{\textstyle \sum \limits_{j=1}^{n}}
\left \vert \bar{h}_{j}k_{j}\right \vert
{\textstyle \int \nolimits_{t-\varepsilon}^{t}}
{\textstyle \int \nolimits_{\tau-D_{i}(\tau)}^{t-D_{i}(t)}}
1\mathrm{d}s\mathrm{d}\tau \\
\text{ \  \  \  \  \  \  \  \  \  \  \ }+\frac{4\kappa \left \vert k_{i}\right \vert
\Delta_{1}}{\varepsilon \left \vert a_{i}\right \vert }\sqrt{%
{\textstyle \sum \limits_{j=1}^{n}}
a_{j}^{2}}\sqrt{%
{\textstyle \sum \limits_{j=1}^{n}}
\frac{k_{j}^{2}}{a_{j}^{2}}}%
{\textstyle \int \nolimits_{t-\varepsilon}^{t}}
{\textstyle \int \nolimits_{\tau-D_{i}(\tau)}^{t-D_{i}(t)}}
1\mathrm{d}s\mathrm{d}\tau \\
\text{ \  \  \  \  \  \  \  \ }\leq \frac{4\left \vert k_{i}\right \vert \Delta_{1}%
}{\varepsilon \left \vert a_{i}\right \vert }\left(
{\textstyle \sum \limits_{j=1}^{n}}
\left \vert \bar{h}_{j}k_{j}\right \vert +\kappa \sqrt{%
{\textstyle \sum \limits_{j=1}^{n}}
\bar{\sigma}_{j}^{2}}\sqrt{%
{\textstyle \sum \limits_{j=1}^{n}}
\frac{k_{j}^{2}}{a_{j}^{2}}}\right) \\
\text{ \  \  \  \  \  \  \  \  \  \  \ }\times%
{\textstyle \int \nolimits_{t-\varepsilon}^{t}}
\left(  t-\tau+2\varepsilon \mu \right)  \mathrm{d}\tau \\
\text{ \  \  \  \  \  \  \  \ }=\frac{\varepsilon \left \vert k_{i}\right \vert
}{\left \vert a_{i}\right \vert }\bar{\Delta}_{2},\text{ }t\geq D_{\mathrm{M}%
}+\varepsilon,\text{ }i=1,\ldots,n
\end{array}
\right.  \label{eq15}%
\end{equation}
with $\bar{\Delta}_{2}$ given in (\ref{eq38}), and%
\begin{equation}
\left.
\begin{array}
[c]{l}%
\left \vert Y_{3i}(t)\right \vert \leq \left \vert C(k_{i})\right \vert \left \Vert
\Delta \mathcal{H}\right \Vert \left \vert \tilde{\vartheta}(t-D_{i}%
(t))\right \vert \\
\text{ \  \  \ }<\left \vert k_{i}\right \vert \kappa \sqrt{%
{\textstyle \sum \limits_{j=1}^{n}}
\bar{\sigma}_{j}^{2}},\text{ }t\geq D_{\mathrm{M}}+\varepsilon,\text{
}i=1,\ldots,n.
\end{array}
\right.  \label{eq44c}%
\end{equation}
Noting the form of $S(t)$ in (\ref{eq55}) with $\omega_{i}$ satisfying
(\ref{eq34a}) and $\left \vert \Delta D(t)\right \vert \leq \rho$ with $\rho$
chosen as (\ref{eq21a}), we find%
\begin{equation}
\left.
\begin{array}
[c]{l}%
\left \vert
{\textstyle \int \nolimits_{t-\Delta D(t)}^{t}}
\dot{S}_{i}(s)\mathrm{d}s\right \vert \leq2\pi \mu i\left \vert a_{i}\right \vert
,\text{ }i=1,\ldots,n\mathbf{,}\\
\left \vert
{\textstyle \int \nolimits_{t-\Delta D(t)}^{t}}
\dot{S}(s)\mathrm{d}s\right \vert \leq2\pi \mu \sqrt{%
{\textstyle \sum \limits_{j=1}^{n}}
j^{2}a_{j}^{2}}.
\end{array}
\right.  \label{eq44a}%
\end{equation}
Employing (\ref{eq11}), (\ref{eq21a}), (\ref{eq25a}), (\ref{eq42}),
(\ref{eq44}) and (\ref{eq44a}), we obtain%
\begin{equation}
\left.
\begin{array}
[c]{l}%
\text{ \  \ }\left \vert k_{i}M_{i}(t)S^{\mathrm{T}}(t)\mathcal{H}%
{\textstyle \int \nolimits_{t-\Delta D(t)}^{t}}
\dot{S}(s)\mathrm{d}s\right \vert \\
\leq \left \vert k_{i}M_{i}(t)\right \vert
{\textstyle \sum \limits_{j=1}^{n}}
\left \vert \bar{h}_{j}\right \vert \left \vert S_{j}(t)\right \vert \left \vert
{\textstyle \int \nolimits_{t-\Delta D(t)}^{t}}
\dot{S}_{j}(s)\mathrm{d}s\right \vert \\
\text{ \  \ }+\left \vert k_{i}M_{i}(t)\right \vert \left \Vert \Delta
\mathcal{H}(t)\right \Vert \left \vert S(t)\right \vert \left \vert
{\textstyle \int \nolimits_{t-\Delta D(t)}^{t}}
\dot{S}(s)\mathrm{d}s\right \vert \\
\leq \frac{4\pi \mu \left \vert k_{i}\right \vert }{\left \vert a_{i}\right \vert }%
{\textstyle \sum \limits_{j=1}^{n}}
\left \vert \bar{h}_{j}\right \vert ja_{j}^{2}+\frac{4\pi \mu \kappa \left \vert
k_{i}\right \vert }{\left \vert a_{i}\right \vert }\sqrt{%
{\textstyle \sum \limits_{j=1}^{n}}
a_{j}^{2}}\sqrt{%
{\textstyle \sum \limits_{j=1}^{n}}
j^{2}a_{j}^{2}},
\end{array}
\right.  \label{eq70}%
\end{equation}%
\begin{equation}
\left.
\begin{array}
[c]{l}%
\text{ \  \ }\left \vert \frac{k_{i}M_{i}(t)}{2}\left(
{\textstyle \int \nolimits_{t-\Delta D(t)}^{t}}
\dot{S}(s)\mathrm{d}s\right)  ^{\mathrm{T}}\mathcal{H}\left(
{\textstyle \int \nolimits_{t-\Delta D(t)}^{t}}
\dot{S}(s)\mathrm{d}s\right)  \right \vert \\
\leq \frac{\left \vert k_{i}M_{i}(t)\right \vert }{2}%
{\textstyle \sum \limits_{j=1}^{n}}
\left \vert \bar{h}_{j}\right \vert \left \vert
{\textstyle \int \nolimits_{t-\Delta D(t)}^{t}}
\dot{S}_{j}(s)\mathrm{d}s\right \vert ^{2}\\
\text{ \  \ }+\frac{\left \vert k_{i}M_{i}(t)\right \vert }{2}\left \Vert
\Delta \mathcal{H}(t)\right \Vert \left \vert
{\textstyle \int \nolimits_{t-\Delta D(t)}^{t}}
\dot{S}(s)\mathrm{d}s\right \vert ^{2}\\
\leq \frac{4\pi^{2}\mu^{2}\left \vert k_{i}\right \vert }{\left \vert
a_{i}\right \vert }%
{\textstyle \sum \limits_{j=1}^{n}}
\left \vert \bar{h}_{j}\right \vert j^{2}a_{j}^{2}+\frac{4\pi^{2}\mu^{2}%
\kappa \left \vert k_{i}\right \vert }{\left \vert a_{i}\right \vert }%
{\textstyle \sum \limits_{j=1}^{n}}
j^{2}a_{j}^{2},
\end{array}
\right.  \label{eq71}%
\end{equation}
and%
\begin{equation}
\left.
\begin{array}
[c]{l}%
\text{ \  \ }\left \vert k_{i}M_{i}(t)\left(
{\textstyle \int \nolimits_{t-\Delta D(t)}^{t}}
\dot{S}(s)\mathrm{d}s\right)  ^{\mathrm{T}}\mathcal{H}\tilde{\vartheta
}(t-D_{i}(t))\right \vert \\
\leq \left \vert k_{i}M_{i}(t)\right \vert
{\textstyle \sum \limits_{j=1}^{n}}
\left \vert \bar{h}_{j}\right \vert \left \vert
{\textstyle \int \nolimits_{t-\Delta D(t)}^{t}}
\dot{S}_{j}(s)\mathrm{d}s\right \vert \left \vert \tilde{\vartheta}_{j}%
(t-D_{i}(t))\right \vert \\
+\left \vert k_{i}M_{i}(t)\right \vert \left \Vert \Delta \mathcal{H}%
(t)\right \Vert \left \vert
{\textstyle \int \nolimits_{t-\Delta D(t)}^{t}}
\dot{S}(s)\mathrm{d}s\right \vert \left \vert \tilde{\vartheta}(t-D_{i}%
(t))\right \vert \\
<\frac{4\pi \mu \left \vert k_{i}\right \vert }{\left \vert a_{i}\right \vert }%
{\textstyle \sum \limits_{j=1}^{n}}
\left \vert \bar{h}_{j}a_{j}\right \vert j\bar{\sigma}_{j}+\frac{4\pi \mu
\kappa \left \vert k_{i}\right \vert }{\left \vert a_{i}\right \vert }\sqrt{%
{\textstyle \sum \limits_{j=1}^{n}}
j^{2}a_{j}^{2}}\sqrt{%
{\textstyle \sum \limits_{j=1}^{n}}
\bar{\sigma}_{j}^{2}}.
\end{array}
\right.  \label{eq72}%
\end{equation}
Then via (\ref{eq70})-(\ref{eq72}), we get from (\ref{eq33a}) that%
\begin{equation}
\left.
\begin{array}
[c]{l}%
\left \vert w_{i}(t)\right \vert <\frac{4\pi \mu \left \vert k_{i}\right \vert
}{\left \vert a_{i}\right \vert }%
{\textstyle \sum \limits_{j=1}^{n}}
\left \vert \bar{h}_{j}\right \vert ja_{j}^{2}+\frac{4\pi \mu \kappa \left \vert
k_{i}\right \vert }{\left \vert a_{i}\right \vert }\sqrt{%
{\textstyle \sum \limits_{j=1}^{n}}
a_{j}^{2}}\sqrt{%
{\textstyle \sum \limits_{j=1}^{n}}
j^{2}a_{j}^{2}}\\
+\frac{4\pi^{2}\mu^{2}\left \vert k_{i}\right \vert }{\left \vert a_{i}%
\right \vert }%
{\textstyle \sum \limits_{j=1}^{n}}
\left \vert \bar{h}_{j}\right \vert j^{2}a_{j}^{2}+\frac{4\pi^{2}\mu^{2}%
\kappa \left \vert k_{i}\right \vert }{\left \vert a_{i}\right \vert }%
{\textstyle \sum \limits_{j=1}^{n}}
j^{2}a_{j}^{2}\\
+\frac{4\pi \mu \left \vert k_{i}\right \vert }{\left \vert a_{i}\right \vert }%
{\textstyle \sum \limits_{j=1}^{n}}
\left \vert \bar{h}_{j}a_{j}\right \vert j\bar{\sigma}_{j}+\frac{4\pi \mu
\kappa \left \vert k_{i}\right \vert }{\left \vert a_{i}\right \vert }\sqrt{%
{\textstyle \sum \limits_{j=1}^{n}}
j^{2}a_{j}^{2}}\sqrt{%
{\textstyle \sum \limits_{j=1}^{n}}
\bar{\sigma}_{j}^{2}}\\
=\frac{\mu \left \vert k_{i}\right \vert }{\left \vert a_{i}\right \vert }%
\Delta_{2},\text{ }i=1,\ldots,n
\end{array}
\right.  \label{eq15a}%
\end{equation}
with $\Delta_{2}$ given in (\ref{eq38}). By using (\ref{eq45})-(\ref{eq44c}),
(\ref{eq15a}) and noting $\bar{w}_{i}(t)$ in (\ref{eq21}), we have%
\begin{equation}
\left.
\begin{array}
[c]{l}%
\left \vert \bar{w}_{i}(t)\right \vert \leq \left \vert k_{i}h_{i}G_{i}%
(t-D_{i}(t))\right \vert +\left \vert Y_{1i}(t)\right \vert \\
\text{ \  \  \  \  \  \  \  \  \  \ }+\left \vert Y_{2i}(t)\right \vert +\left \vert
Y_{3i}(t)\right \vert +\left \vert w_{i}(t)\right \vert \\
\text{ \  \  \  \  \  \  \ }<W_{i}(\varepsilon,\mu,\kappa),\text{ }t\geq
D_{\mathrm{M}}+\varepsilon,\text{ }i=1,\ldots,n
\end{array}
\right.  \label{eq19}%
\end{equation}
with $W_{i}(\varepsilon,\mu,\kappa)$ given by (\ref{eq24b}).

In addition, via (\ref{eq52a}), (\ref{eq64}) and (\ref{eq13}) we find for
$i=1,\ldots,n\mathbf{,}$%
\begin{equation}
\left.
\begin{array}
[c]{l}%
\left \vert z_{i}(t)\right \vert \leq \left \vert \tilde{\vartheta}_{i}%
(t)\right \vert +|G_{i}(t)|\\
<\left \vert \tilde{\vartheta}_{i}(D_{\mathrm{M}})\right \vert +\frac
{3\varepsilon \left \vert k_{i}\right \vert \Delta_{1}}{\left \vert a_{i}%
\right \vert },\text{ }t\in \lbrack \varepsilon,D_{\mathrm{M}}+\varepsilon],
\end{array}
\right.  \label{eq56}%
\end{equation}
where we have used $\tilde{\vartheta}_{i}(t)=\tilde{\vartheta}_{i}%
(D_{\mathrm{M}}),$ $t\in \lbrack0,D_{\mathrm{M}}].$

\textit{Proof of part B.} Define $X_{i}(t,s)$ as the solution of the following
homogeneous equation%
\[
\left.
\begin{array}
[c]{l}%
\dot{z}_{i}(t)=k_{i}\bar{h}_{i}z(t-D_{i}(t)),\text{ }t\geq s,\text{
}i=1,\ldots,n\mathbf{,}\\
z_{i}(t)=0,t<s,\text{ }z_{i}(s)=1,\text{ }s\geq0.
\end{array}
\right.
\]
By using Lemma \ref{lemma1}, under the condition $\Phi_{1}^{i}\leq0$
$(i=1,\ldots,n)$ in (\ref{eq24a}), there hold for $i=1,\ldots,n\mathbf{,}$%
\begin{equation}
\left.  0<X_{i}(t,s)\leq \left \{
\begin{array}
[c]{ll}%
1, & s\leq t\leq s+D_{i\mathrm{M}},\\
\mathrm{e}^{-\left \vert k_{i}\bar{h}_{i}\right \vert \left(  t-s-D_{i\mathrm{M}%
}\left(  \varepsilon \right)  \right)  }, & t\geq s+D_{i\mathrm{M}}%
\end{array}
\right.  \right.  \label{eq69}%
\end{equation}
with $D_{i\mathrm{M}}$ $(i=1,\ldots,n)$ given in (\ref{eq21b}). By using
(\ref{eq77}) in Lemma \ref{lemma1} for (\ref{eq21}) we further have%
\[
\left.
\begin{array}
[c]{l}%
z_{i}(t)=X_{i}(t,D_{\mathrm{M}}+\varepsilon)z_{i}(D_{\mathrm{M}}%
+\varepsilon)\\
\text{ \  \  \  \  \  \  \ }+\int \nolimits_{D_{\mathrm{M}}+\varepsilon}^{t}%
X_{i}(t,s)k_{i}\bar{h}_{i}\varphi_{i}(s-D_{i}(s))\mathrm{d}s\\
\text{ \  \  \  \  \  \  \ }+\int \nolimits_{D_{\mathrm{M}}+\varepsilon}^{t}%
X_{i}(t,s)\bar{w}_{i}(s)\mathrm{d}s,
\end{array}
\right.
\]
where $\varphi_{i}(s-D_{i}(s))=0$ if $s-D_{i}(s)>D_{\mathrm{M}}+\varepsilon$
and $\varphi_{i}(s-D_{i}(s))=z_{i}(s-D_{i}(s))$ if $\varepsilon \leq
s-D_{i}(s)\leq D_{\mathrm{M}}+\varepsilon.$ Then it follows from (\ref{eq19})
and (\ref{eq56}) that%
\begin{equation}
\left.
\begin{array}
[c]{l}%
\left \vert z_{i}(t)\right \vert \leq \left \vert X_{i}(t,D_{\mathrm{M}%
}+\varepsilon)\right \vert \left \vert z_{i}(D_{\mathrm{M}}+\varepsilon
)\right \vert \\
\text{ \  \  \  \  \  \  \  \  \ }+\left \vert k_{i}\bar{h}_{i}\right \vert
\int \nolimits_{D_{\mathrm{M}}+\varepsilon}^{t}\left \vert X_{i}(t,s)\right \vert
\left \vert \varphi_{i}(s-D_{i}(s))\right \vert \mathrm{d}s\\
\text{ \  \  \  \  \  \  \  \  \ }+\int \nolimits_{D_{\mathrm{M}}+\varepsilon}%
^{t}\left \vert X_{i}(t,s)\right \vert \left \vert \bar{w}_{i}(s)\right \vert
\mathrm{d}s\\
\text{ \  \  \  \  \  \ }<\left[  \left \vert \tilde{\vartheta}_{i}(D_{\mathrm{M}%
})\right \vert +\frac{3\varepsilon \left \vert k_{i}\right \vert \Delta_{1}%
}{\left \vert a_{i}\right \vert }\right]  \left \vert X_{i}(t,D_{\mathrm{M}%
}+\varepsilon)\right \vert \\
\text{ \  \  \  \  \  \  \  \  \ }+\left \vert k_{i}\bar{h}_{i}\right \vert \left[
\left \vert \tilde{\vartheta}_{i}(D_{\mathrm{M}})\right \vert +\frac
{3\varepsilon \left \vert k_{i}\right \vert \Delta_{1}}{\left \vert a_{i}%
\right \vert }\right] \\
\text{ \  \  \  \  \  \  \  \  \ }\times \int \nolimits_{D_{\mathrm{M}}+\varepsilon
}^{D_{\mathrm{M}}+D_{i\mathrm{M}}+\varepsilon}\left \vert X_{i}(t,s)\right \vert
\mathrm{d}s\\
\text{ \  \  \  \  \  \  \  \  \ }+W_{i}(\varepsilon,\mu,\kappa)\int
\nolimits_{D_{\mathrm{M}}+\varepsilon}^{t}\left \vert X_{i}(t,s)\right \vert
\mathrm{d}s
\end{array}
\right.  \label{eq23}%
\end{equation}
When $t\in \lbrack D_{\mathrm{M}}+\varepsilon,D_{\mathrm{M}}+D_{i\mathrm{M}%
}+\varepsilon],$ by using (\ref{eq69}), inequality (\ref{eq23}) can be
continued as%
\[
\left.
\begin{array}
[c]{l}%
\left \vert z_{i}(t)\right \vert <\left(  1+\left \vert k_{i}\bar{h}%
_{i}\right \vert D_{i\mathrm{M}}\right)  \left(  \left \vert \tilde{\vartheta
}_{i}(D_{\mathrm{M}})\right \vert +\frac{3\varepsilon \left \vert k_{i}%
\right \vert \Delta_{1}}{\left \vert a_{i}\right \vert }\right) \\
\text{ \  \  \  \  \  \  \  \ }+D_{i\mathrm{M}})W_{i}(\varepsilon,\mu,\kappa),
\end{array}
\right.
\]
by which, (\ref{eq52a}) and (\ref{eq13}), we further have%
\begin{equation}
\left.
\begin{array}
[c]{l}%
\left \vert \tilde{\vartheta}_{i}(t)\right \vert <\left(  1+\left \vert k_{i}%
\bar{h}_{i}\right \vert D_{i\mathrm{M}}\right)  \left(  \left \vert
\tilde{\vartheta}_{i}(D_{\mathrm{M}})\right \vert +\frac{3\varepsilon \left \vert
k_{i}\right \vert \Delta_{1}}{\left \vert a_{i}\right \vert }\right) \\
\text{ \  \  \  \  \  \  \  \  \ }+D_{i\mathrm{M}}W_{i}(\varepsilon,\mu,\kappa
)+\frac{\varepsilon \left \vert k_{i}\right \vert \Delta_{1}}{\left \vert
a_{i}\right \vert },
\end{array}
\right.  \label{eq29}%
\end{equation}
which implies the second inequality in (\ref{eq65a}) due to $\Phi_{2}^{i}<0$
in (\ref{eq24a}) since $\mathrm{e}^{\left \vert k_{i}\bar{h}_{i}\right \vert
D_{i\mathrm{M}}}\geq1+\left \vert k_{i}\bar{h}_{i}\right \vert D_{i\mathrm{M}}.$

When $t\geq D_{\mathrm{M}}+D_{i\mathrm{M}}+\varepsilon,$ by using
(\ref{eq69}), inequality (\ref{eq23}) can be continued as%
\[
\left.
\begin{array}
[c]{l}%
\left \vert z_{i}(t)\right \vert <\mathrm{e}^{-\left \vert k_{i}\bar{h}%
_{i}\right \vert \left(  t-D_{\mathrm{M}}-D_{i\mathrm{M}}-\varepsilon \right)
}\left[  \left \vert \tilde{\vartheta}_{i}(D_{\mathrm{M}})\right \vert
+\frac{3\varepsilon \left \vert k_{i}\right \vert \Delta_{1}}{\left \vert
a_{i}\right \vert }\right] \\
\text{ \  \  \  \  \  \  \  \ }+\left \vert k_{i}\bar{h}_{i}\right \vert \left[
\left \vert \tilde{\vartheta}_{i}(D_{\mathrm{M}})\right \vert +\frac
{3\varepsilon \left \vert k_{i}\right \vert \Delta_{1}}{\left \vert a_{i}%
\right \vert }\right] \\
\text{ \  \  \  \  \  \  \  \ }\times \int \nolimits_{D_{\mathrm{M}}+\varepsilon
}^{D_{\mathrm{M}}+D_{i\mathrm{M}}+\varepsilon}\mathrm{e}^{-\left \vert
k_{i}\bar{h}_{i}\right \vert \left(  t-s-D_{i\mathrm{M}}\right)  }\mathrm{d}s\\
\text{ \  \  \  \  \  \  \  \ }+W_{i}(\varepsilon,\mu,\kappa)\int
\nolimits_{D_{\mathrm{M}}+\varepsilon}^{t}\mathrm{e}^{-\left \vert k_{i}\bar
{h}_{i}\right \vert \left(  t-s-D_{i\mathrm{M}}\right)  }\mathrm{d}s\\
\text{ \  \  \  \  \ }=\mathrm{e}^{-\left \vert k_{i}\bar{h}_{i}\right \vert \left(
t-D_{\mathrm{M}}-D_{i\mathrm{M}}-\varepsilon \right)  }\left[  \left \vert
\tilde{\vartheta}_{i}(D_{\mathrm{M}})\right \vert +\frac{3\varepsilon \left \vert
k_{i}\right \vert \Delta_{1}}{\left \vert a_{i}\right \vert }\right] \\
\text{ \  \  \  \  \  \  \  \ }+\mathrm{e}^{-\left \vert k_{i}\bar{h}_{i}\right \vert
\left(  t-D_{\mathrm{M}}-D_{i\mathrm{M}}-\varepsilon \right)  }\left(
\mathrm{e}^{\left \vert k_{i}\bar{h}_{i}\right \vert D_{i\mathrm{M}}}-1\right)
\\
\text{ \  \  \  \  \  \  \  \ }\times \left[  \left \vert \tilde{\vartheta}%
_{i}(D_{\mathrm{M}})\right \vert +\frac{3\varepsilon \left \vert k_{i}\right \vert
\Delta_{1}}{\left \vert a_{i}\right \vert }\right] \\
\text{ \  \  \  \  \  \  \  \ }+\frac{W_{i}(\varepsilon,\mu,\kappa)}{\left \vert
k_{i}\bar{h}_{i}\right \vert }\left[  \mathrm{e}^{\left \vert k_{i}\bar{h}%
_{i}\right \vert D_{i\mathrm{M}}}-\mathrm{e}^{-\left \vert k_{i}\bar{h}%
_{i}\right \vert \left(  t-D_{M}-D_{i\mathrm{M}}-\varepsilon \right)  }\right]
\\
\text{ \  \  \  \  \ }\leq \mathrm{e}^{-\left \vert k_{i}\bar{h}_{i}\right \vert
\left(  t-D_{\mathrm{M}}-2D_{i\mathrm{M}}-\varepsilon \right)  }\left[
\left \vert \tilde{\vartheta}_{i}(D_{\mathrm{M}})\right \vert +\frac
{3\varepsilon \left \vert k_{i}\right \vert \Delta_{1}}{\left \vert a_{i}%
\right \vert }\right] \\
\text{ \  \  \  \  \  \  \  \ }+\frac{\mathrm{e}^{\left \vert k_{i}\bar{h}%
_{i}\right \vert D_{i\mathrm{M}}}}{\left \vert k_{i}\bar{h}_{i}\right \vert
}W_{i}(\varepsilon,\mu,\kappa).
\end{array}
\right.
\]
The latter together with (\ref{eq52a}) and (\ref{eq13}) yield%
\[
\left.
\begin{array}
[c]{l}%
\left \vert \tilde{\vartheta}_{i}(t)\right \vert <\mathrm{e}^{-\left \vert
k_{i}\bar{h}_{i}\right \vert \left(  t-D_{\mathrm{M}}-2D_{i\mathrm{M}%
}-\varepsilon \right)  }\left[  \left \vert \tilde{\vartheta}_{i}(D_{\mathrm{M}%
})\right \vert +\frac{3\varepsilon \left \vert k_{i}\right \vert \Delta_{1}%
}{\left \vert a_{i}\right \vert }\right] \\
\text{ \  \  \  \  \  \  \  \  \  \ }+\frac{\mathrm{e}^{\left \vert k_{i}\bar{h}%
_{i}\right \vert D_{i\mathrm{M}}}}{\left \vert k_{i}\bar{h}_{i}\right \vert
}W_{i}(\varepsilon,\mu,\kappa)+\frac{\varepsilon \left \vert k_{i}\right \vert
\Delta_{1}}{\left \vert a_{i}\right \vert },
\end{array}
\right.
\]
which implies the third inequality in (\ref{eq65a}) due to $\Phi_{2}^{i}<0$ in
(\ref{eq24a}).

\textit{Proof of part C.} we show that the conditions in (\ref{eq24a})
guarantee that bounds in (\ref{eq25}) hold.

\textbf{(i)} When $t\in \lbrack D_{\mathrm{M}},D_{\mathrm{M}}+\varepsilon],$
since $|\tilde{\vartheta}_{i}(D_{\mathrm{M}})|\leq \bar{\sigma}_{0i}%
<\bar{\sigma}_{i}$ and $\tilde{\vartheta}_{i}(t)$ is continuous in time,
(\ref{eq25}) holds for some $t>D_{\mathrm{M}}.$ We assume by contradiction
that for some $t\in(D_{\mathrm{M}},D_{\mathrm{M}}+\varepsilon]$ the formula
(\ref{eq25}) does not hold for some $i$. Namely, there exists the smallest
time instant $t^{\ast}\in(D_{\mathrm{M}},D_{\mathrm{M}}+\varepsilon]$ such
that $|\tilde{\vartheta}_{i}(t^{\ast})|=\bar{\sigma}_{i},$ $|\tilde{\vartheta
}_{i}(t)|<\bar{\sigma}_{i},$ $t\in \lbrack D_{\mathrm{M}},t^{\ast}).$ Then by
employing inequalities of (\ref{eq25})-(\ref{eq28}), we obtain%
\[
\left.  \left \vert \dot{\tilde{\vartheta}}_{i}(t)\right \vert \leq \frac
{2k_{i}\Delta_{1}}{\left \vert a_{i}\right \vert },\text{ }t\in \lbrack
D_{\mathrm{M}},t^{\ast}].\right.
\]
As a result,%
\[
\left.
\begin{array}
[c]{l}%
\left \vert \tilde{\vartheta}_{i}(t)\right \vert =\left \vert \tilde{\vartheta
}_{i}(D_{\mathrm{M}})+%
{\textstyle \int \nolimits_{D_{M}}^{t}}
\dot{\tilde{\vartheta}}_{i}(s)\mathrm{d}s\right \vert \\
\text{ \  \  \  \  \  \  \  \  \ }\leq \left \vert \tilde{\vartheta}_{i}(D_{\mathrm{M}%
})\right \vert +\frac{2\varepsilon \left \vert k_{i}\right \vert \Delta_{1}%
}{\left \vert a_{i}\right \vert },\text{ }t\in \lbrack D_{\mathrm{M}},t^{\ast}].
\end{array}
\right.
\]
Furthermore, the feasibility of $\Phi_{2}^{i}<0$ in (\ref{eq24a}) ensures that%
\[
\left.  \left \vert \tilde{\vartheta}_{i}(t^{\ast})\right \vert \leq \bar{\sigma
}_{0i}+\frac{2\varepsilon^{\ast}\left \vert k_{i}\right \vert \Delta_{1}%
}{\left \vert a_{i}\right \vert }<\bar{\sigma}_{i}.\right.
\]
This contradicts to the definition of $t^{\ast}$ such that $|\tilde{\vartheta
}_{i}(t^{\ast})|=\bar{\sigma}_{i}.$ Hence (\ref{eq25}) holds for $t\in \lbrack
D_{\mathrm{M}},D_{\mathrm{M}}+\varepsilon].$

\textbf{(ii)} When $t\in \lbrack D_{\mathrm{M}}+\varepsilon,2D_{\mathrm{M}%
}+\varepsilon],$ since $|\tilde{\vartheta}_{i}(D_{\mathrm{M}}+\varepsilon
)|<\bar{\sigma}_{i}$ as shown in \textbf{(i) }and $\tilde{\vartheta}_{i}(t)$
is continuous in time, (\ref{eq25}) holds for some $t>D_{\mathrm{M}%
}+\varepsilon.$ We assume by contradiction that for some $t\in(D_{\mathrm{M}%
}+\varepsilon,2D_{\mathrm{M}}+\varepsilon]$ the formula (\ref{eq25}) does not
hold for some $i$. Namely, there exists the smallest time instant $t^{\ast}%
\in(D_{\mathrm{M}}+\varepsilon,2D_{\mathrm{M}}+\varepsilon]$ such that
$|\tilde{\vartheta}_{i}(t^{\ast})|=\bar{\sigma}_{i},$ $|\tilde{\vartheta}%
_{i}(t)|<\bar{\sigma}_{i},$ $t\in \lbrack D_{\mathrm{M}}+\varepsilon,t^{\ast
}).$ This together with $|\tilde{\vartheta}_{i}(t)|<\bar{\sigma}_{i},t\in \lbrack
D_{\mathrm{M}},D_{\mathrm{M}}+\varepsilon]$ in \textbf{(i)} give
$|\tilde{\vartheta}_{i}(t)|\leq \bar{\sigma}_{i},$ $t\in \lbrack D_{\mathrm{M}%
},t^{\ast}].$ Similar to the proofs in parts (A) and (B), under the condition
$\Phi_{1}^{i}\leq0$ in (\ref{eq24a}), we finally arrive at (\ref{eq29}) in its
non-strict version for $t\in \lbrack D_{\mathrm{M}}+\varepsilon,t^{\ast}].$
Moreover, the feasibility of $\Phi_{2}^{i}<0$ in (\ref{eq24a}) ensures that
$|\tilde{\vartheta}_{i}(t^{\ast})|<\bar{\sigma}_{i}.$ This contradicts to the
definition of $t^{\ast}$ such that $|\tilde{\vartheta}_{i}(t^{\ast}%
)|=\bar{\sigma}_{i}.$ Hence (\ref{eq25}) holds for $t\in \lbrack D_{\mathrm{M}%
}+\varepsilon,2D_{\mathrm{M}}+\varepsilon].$

\textbf{(iii)} By arguments of (\textbf{ii}) it can be proved that
(\ref{eq24a}) guarantees (\ref{eq25}) for all $t\geq2D_{\mathrm{M}%
}+\varepsilon$, which finishes the proof.

\end{document}